\documentclass{ijuc}

\usepackage[english]{babel}
\usepackage{amsthm,amsmath,amscd,amsfonts,amssymb,mathrsfs,stmaryrd,mathpazo,subfig,calc,verbatim,color,qtree} 
\usepackage[pdftex]{graphicx}

\newcommand{\F}{\mathbb{F}}
\newcommand{\N}{\mathbb{N}}
\newcommand{\R}{\mathbb{R}}
\newcommand{\Z}{\mathbb{Z}}
\newcommand{\Id}{\mathbb{1}}

\newcommand{\abs}[1]{\left|#1\right|}
\newcommand{\acco}[1]{\left\{#1\right\}}
\newcommand{\cro}[1]{\left[#1\right]}
\newcommand{\doublecro}[1]{\left\llbracket #1 \right\rrbracket } 
\newcommand{\I}{\mathbb{I}}  
\newcommand{\floor}[1]{\left\lfloor #1 \right\rfloor}
\newcommand{\Hom}{\operatorname{Hom}}

\newcommand{\incl}{\subseteq}

\newcommand{\joliC}{\mathscr{C}}
\newcommand{\mochec}{\mathfrak{c}}

\newcommand{\joliD}{\mathscr{D}}

\newcommand{\joliM}{\mathscr{M}}

\newcommand{\joliP}{\mathscr{P}}
\newcommand{\joliS}{\mathscr{S}}

\newcommand{\pa}[1]{\left(#1\right)}

\newcommand{\ud}[1]{\rotatebox[origin=c]{180}{$#1$}}
\newcommand{\tr}{\operatorname{tr}}

\newcommand{\ceil}[1]{\left\lceil#1 \right\rceil}
\newcommand{\polyI}{\mathcal{I}}

\newcommand{\sca}{\Theta}

\newtheorem{prop}{Proposition}

\newtheorem{thm}{Theorem}

\begin{document}

\author{Johannes~G{\"u}tschow\inst{1} \and Vincent~Nesme\inst{2}\email{vnesme@gmail.com} \and Reinhard~F.~Werner\inst{1}}
\title{The fractal structure of cellular automata on Abelian groups}

\institute{Institut f{\"u}r Theoretische Physik, Universit\"at Hannover
\and
Quantum information theory, Universit\"at Potsdam
}

\maketitle

\begin{abstract}

It is well-known that the spacetime diagrams of some cellular automata have a fractal structure: for instance Pascal's triangle modulo 2 generates a Sierpinski triangle. Explaining the fractal structure of the spacetime diagrams of cellular automata is a much explored topic, but virtually all of the results revolve around a special class of automata, whose typical features include irreversibility, an alphabet with a ring structure, a global evolution that is a ring homomorphism, and a property known as (weakly) $p$-Fermat.  
The class of automata that we study in this article has none of these properties. Their cell structure is weaker, as it does not come with a multiplication, and they are far from being $p$-Fermat, even weakly. However, they do produce fractal spacetime diagrams, and we explain why and how.

\end{abstract}

\section*{Introduction}

The fractal structure of cellular automata (CAs) has been a topic of interest for several decades. In many works on linear CAs, the authors present ways to calculate the fractal dimension or to predict the state of an arbitrary cell at an arbitrary time step, with much lower complexity than by running the CA step by step; However, their notions of linearity are quite different. Often only CAs that use states in $\Z_p$ are studied.\footnote{We use the simple notation $\Z_d$ for the cyclic group of order $d$, instead of $\Z/d\Z$, as we are concerned with finite groups only.} 
Other approaches are more general, but still make certain assumptions on the time evolution or the underlying structure of the CA. In this work we try to loosen these restrictions as far as possible. We consider one-dimensional linear CAs whose alphabet is an abelian group. We show how they can be described by $n\times n$ matrices with polynomial entries and use this description to derive a recursion relation for the iterations of the CA. This recursion relation enables us to formulate the evolution of the spacetime diagram as a matrix substitution system, which in turn gives us the means to calculate the fractal dimension of the spacetime diagram.
  
  Most of the methods we employ are commonly used in the study of CAs. To prove that the spacetime diagram converges we use elementary graph theory and study the graph associated with the matrix substitution system of the CA. The matrix substitution system is obtained by a recursion formula for the iterations of the CA and grouping of cells.
  
  Our interest in the fractal structure of CAs on abelian groups stems from our study of Clifford quantum cellular automata (CQCAs) \cite{SchlingemannCQCA}. We first noticed the self-similar structure while studying their long time behaviour \cite{Guetschow2009,Guetschow2009a}.  A CQCA maps Pauli matrices to tensor products of Pauli matrices times a phase. If we neglect the phase, we can identify the Pauli matrices $X$, $Y$, $Z$ with the elements of $\Z^2_2$ via the mapping $X\mapsto\tbinom{1}{0},\,Y\mapsto\tbinom{1}{1},\,Z\mapsto\tbinom{0}{1},\,\Id\mapsto\tbinom{0}{0}$. Using this mapping we can simulate CQCAs with linear CAs on the alphabet $\Z_2^2$.  The CAs corresponding to CQCAs have to fulfill rather strong conditions: they have to be reversible and preserve a symplectic form which encodes the commutation relations of the Pauli matrices \cite{SchlingemannCQCA}. While our analysis is now much more general, our main example $\sca$, whose spacetime diagram is shown in figure~\ref{fig:theta}, is the classical counterpart of a CQCA.

  \begin{figure}[htbp]
    \centering
    \subfloat[Time evolution of $\sca$ with scalars in $\Z_2$.]{\label{fig:theta}\includegraphics[width=\textwidth/2-3mm]{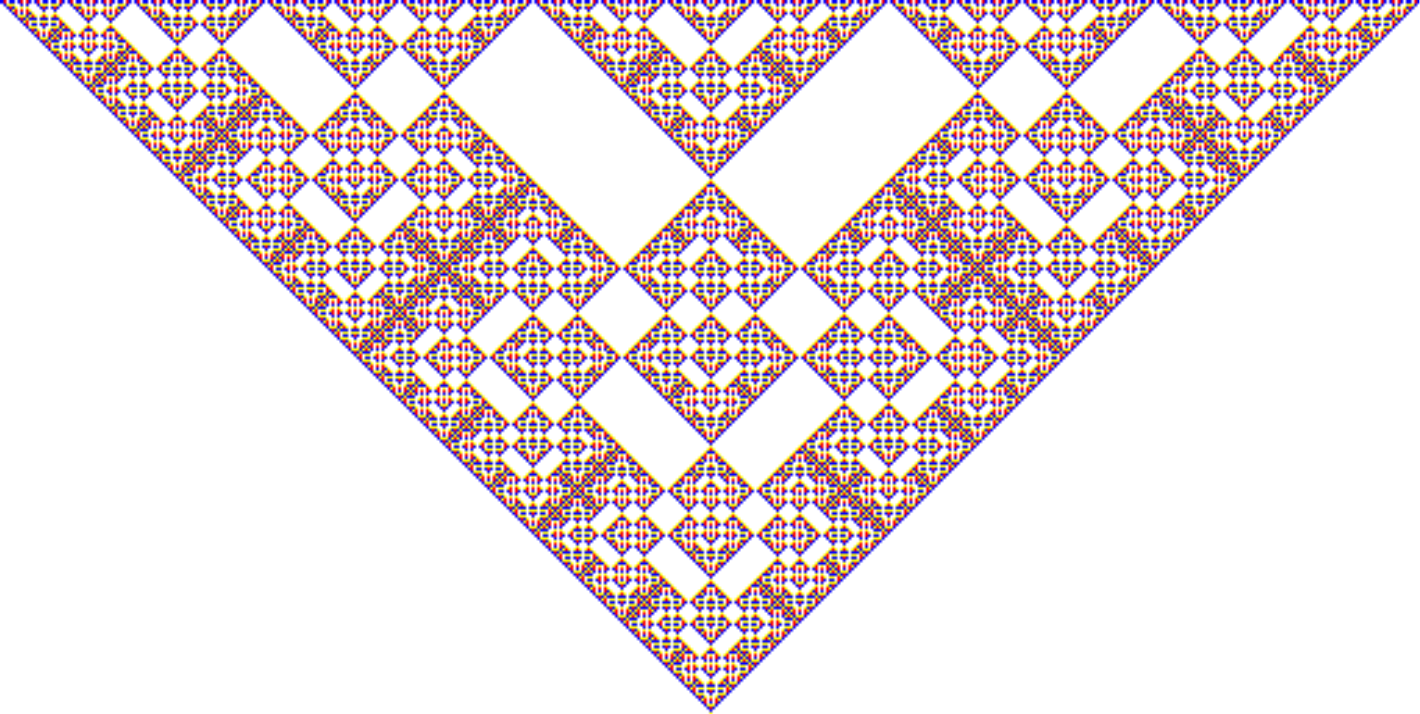}}
    \hskip 3mm 
    \subfloat[Time evolution of $\sca$ with scalars extended to $\Z_4$.]{\label{fig:theta_p4}\includegraphics[width=\textwidth/2-3mm]{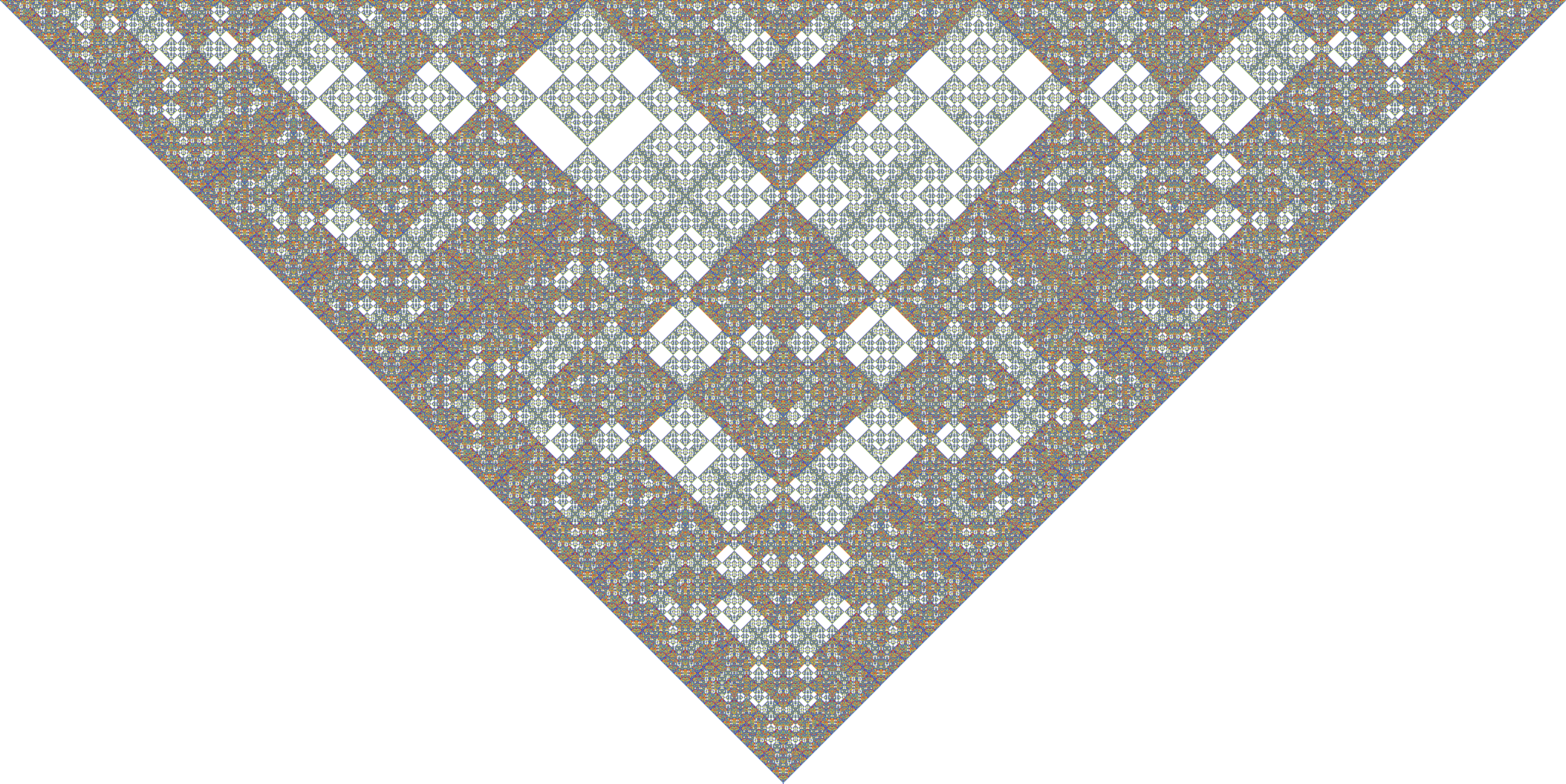}}
    \caption{(a) is a projection of (b) induced by $\Z_4^2\twoheadrightarrow\Z_2^2$.}
  \end{figure}
      
  Our paper is organized as follows: in section~\ref{section_definitions} we give our definition of a linear cellular automaton, introduce the formalism we will be working with, and state the main result: every linear cellular automaton has a fractal structure.  We will also introduce the example $\sca$ which will be the focal point of this article. In section~\ref{section_special_scheme}, we will give an intuitive idea as to why the spacetime diagram of $\Theta$ exhibits a fractal structure.  We will then proceed, in section~\ref{section_recursion}, to expose an algorithm taking as input the local transition rule and outputting a description of the spacetime diagram.  This allows us to compute salient features of these fractals, such as their fractal dimension and their average color.

\section{Definitions}\label{section_definitions}
\subsection{Generalities on summable automata}
\subsubsection{Monoids}
We want to discuss ``summable automata'', for which it makes sense to talk about the influence of a single cell on every other cell, and where the global transition function can be reconstructed by ``summing'' all these influences.  So, if $\Sigma$ denotes the alphabet, instead of the usual local transition function $\Sigma^I\to\Sigma$, a summable automaton is naturally defined by a function $\Sigma\to\Sigma^I$.  What is then the minimal structure on $\Sigma$ that would make such a definition work?  These influences have to be ``summed'', so we need an operation on $\Sigma$.  Since the strip is infinite, an infinitary operation would do, but that wouldn't give us much to work with.  Instead, it seems reasonable to consider a binary operation $+$.  In the same spirit, when we think of the superposition of influences coming from each cell, 
no notion of order between the cells is involved; even if in the one-dimensional case a natural order can be put on the cells, it would be less than clear what to do in higher dimensions.  We require therefore that $+$ be associative and commutative.  The last requirement comes from the fact that, given only the global transition function, we want to be able to isolate the influence of one cell; that is why we demand that $+$ have an identity element, which makes now $\pa{\Sigma,+}$ an Abelian monoid.  Of course, in order for all of this to be relevant, the transition function has to be a morphism.

Let $I$ be some finite subset of $\Z$ and $f$ a morphism from $\Sigma$ to $\Sigma^I$.  From $f$ one can define the global transition function as an endomorphism $F$ of $\Sigma^{\Z}$ by 
\begin{equation}\label{F_and_f}
F:\pa{\begin{array}{rcl}
\Sigma^{\Z}&\to&\Sigma^{\Z} \\
r=\pa{r_n}_{n\in \Z}&\mapsto &\pa{\displaystyle\sum_{i\in I} f(r_{n-i})_{i} }_{n\in \Z}
\end{array}}.
\end{equation}

Let $\sigma$ be the right shift on $\Sigma_\Z$, i.e. $\sigma(r)_n=r_{n-1}$.  We have $F\circ \sigma=\sigma\circ F$, which means $F$ is translation invariant.  Also, $F(r)_n$ depends only on the values $r_{n-i}$ for $i\in I$; since $I$ is finite, $F$ is a one-dimensional cellular automaton on the alphabet $\Sigma$, with neighborhood included in $-I$.  Conversely, if $F$ is an endomorphism of $\Sigma^\Z$ defining a cellular automaton over the alphabet $\Sigma$, then one can choose a neighborhood $I$, and define, for $i\in I$,
\begin{equation}
f(s)_i= F(\bar s)_{i},
\end{equation}
where $\bar s$ is the word of $\Sigma_\Z$ defined by ${\bar s}_n=\left\{\begin{array}{ll}s&\text{if $n=0$} \\e &\text{otherwise} \end{array} \right.$, $e$ denoting the neutral element of $\Sigma$.

\subsubsection{Groups}\label{section_groups}
We will now consider the case when $\Sigma$ is a (finite abelian) group.  For $p$ prime, let $\Sigma_p$ be the subgroup of $\Sigma$ of elements of order a power of $p$; then $\Sigma$ is isomorphic to $\prod\limits_p \Sigma_p$, and every endomorphism of $\Sigma^\Z$ factorises into a product of endomorphisms of the $\Sigma_p^\Z$'s.  It is therefore enough to study the case of the  (abelian) $p$-groups: let us assume $\Sigma$ is a $p$-group.

It is a well-known fact (see for instance section I-8 of \cite{lang}) that $\Sigma$ is isomorphic to $\Z_{p^{k_1}}\times \Z_{p^{k_2}}\times\cdots\times \Z_{p^{k_d}}$ with $k_d \geq k_{d-1} \geq \ldots \geq k_1= k$.  Now let us consider an endomorphism $\alpha$ of $\Sigma$ and let $e_j$ denote $(0,\ldots,0,1,0,\ldots,0)$, where the $1$ lies in position $j$. 
When $i\geq j$, there is a natural embedding $s_{i,j}$ of $\Z_{p^{k_i}}$ into $\Z_{p^{k_j}}$, namely the multiplication by $p^{k_j-k_i}\in\N$.  Since $e_j$ has order $p^{k_j}$, $\alpha(e_j)_i\in\Z_{p^{k_i}}$ has to be in the image of $s_{j,i}$ when $i\leq j$.  We can therefore associate to $\alpha$ the endomorphism of $\Z_{p^k}^{d}$ given by the matrix $A(\alpha)\in \joliM_d(\Z_{p^k})$ defined by $A(\alpha)_{i,j} =p^{k_j-k_i} \alpha\pa{e_j}_i$.

For instance, if $G$ is $\Z_{32}\times\Z_4\times\Z_2$, and $\alpha$ is defined by $\alpha(1,0,0)=(3,2,1)$, $\alpha(0,1,0)=(24,0,1)$ and $\alpha(0,0,1)=(16,2,0)$, then the corresponding matrix of $\joliM_3\pa{\Z_{32}}$ would be 
\begin{equation*}
  A(\alpha)=\pa{\begin{array}{ccc}
3 & 3 & 1 \\ 16 & 0 & 1 \\ 16 & 2 & 0 
\end{array}}.
\end{equation*}

Let us give a summary of the construction we have just exposed.

\begin{prop}
For every finite abelian $p$-group $G$ and endomorphism $\alpha$ of $G$, there are positive integers $k$ and $d$, an embedding $s$ of $G$ into $\Z_{p^k}^d$, and an endomorphism $A(\alpha)$ of $\Z_{p^k}^d$ such that the following diagram commutes:

\begin{equation}
\begin{CD}
G @>\alpha>> G \\
@VVsV @VVsV \\
\Z_{p^k}^d @>A(\alpha)>> \Z_{p^k}^d
\end{CD}
\end{equation}
\end{prop}

This implies that to study the behaviour of CA on abelian groups, it is enough to study the case where these groups are of the form $\Z_{p^k}^d$.

\subsubsection{$R$-modules}
We will actually consider the more general case where $R$ is a finite commutative ring, and $\Sigma$ is a free $R$-module of dimension $d$, i.e. isomorphic to $R^d$.  The first reason for doing so is that it does not complicate the mathematics.  It will also appear more efficient to understand, for instance, $\F_{2^4}$ as a $1$-dimensional vector space over itself than as a $4$-dimensional vector space over $\F_2$: the former simply bears more information, and therefore implies more restrictions on the form of a CA, so that more can be deduced.

For any ring $B$, $B\cro{u,u^{-1}}$ denotes the ring of Laurent polynomials over $B$; it is the ring of linear combinations of integer powers (negative as well as nonnegative) of the unknown $u$.  Applying this to $B =\Hom_R\pa{\Sigma}$, we can associate to the function $f$ the Laurent polynomial $\tau(f)\in \Hom_R\pa{\Sigma}\cro{u,u^{-1}}$ defined by
\begin{equation}
  \label{eq:fourier}
  \tau(f)=\sum_{n\in\Z}f(\cdot)_n u^n.
\end{equation}

$\tau$ is an isomorphism of $R$-algebras between the linear cellular automata on the alphabet $\Sigma$ with internal composition rules $(+,\circ)$ and $\Hom_R\pa{\Sigma}\cro{u,u^{-1}}$, which can be identified with $\joliM_d\pa{R\cro{u,u^{-1}}}$ because $\Sigma\simeq R^d$; we are going to think and work in this former algebra, so from now on a linear cellular automaton $T=\tau(f)$ will be for us an element of $\joliM_d\pa{R\cro{u,u^{-1}}}$.

\subsection{Related work}\label{sec_related_work}
  Many papers have been published about the fractal structure of cellular automata spacetime diagrams.  We give here a short review and point out the differences to our approach.  When we mention $d$ and $k$ we are referring to $\joliM_d\pa{\Z_k[u,u^{-1}]}$.
  \begin{description}
    \item [\cite{Willson1987}] In this work, Willson considers the case $d=1$, $k=2$.  In order to determine the fractal dimension of the spacetime diagram, he analyses how blocks of length $n$ in the configuration of time step $t$ are mapped to such blocks in steps $2t$ and $2t+1$, a technique we also use in section~\ref{proof_sec}.

    \item [\cite{96229,selfsimilarity}] Takahashi generalizes Willson's work to the case $d=1$, with no restriction on the value of $k$.
    \item [\cite{Haeseler1993,Haeseler2001all}] Haeseler, Peitgen and Skordev study the fractal time evolution of CAs with special scaling properties, the weakest of them being ``weakly $p$-Fermat'', where $p$ is some integer, which includes the case $d=1$, $k=p$. Let us briefly introduce the $p$-Fermat property and show why the CAs that we study do not have to be $p$-Fermat. Let $\pi_p$ be the scaling map
    \begin{equation}
      \pi_p(\xi)_x=\left\{\begin{array}{ll}
                         \xi_y&\text{if}\,x=py\\
                         e&\text{otherwise} 
                       \end{array}\right. .
    \end{equation}   
    A CA $T$ is weakly $p$-Fermat if for all $s\in\Sigma$, $n\in\N$ and $x\in\Z$, $T^{np}(\bar s)_x=e\Leftrightarrow\pi_p T^n(\bar s)_x=e$.

    Let us now consider
    \begin{equation}
      \label{eq:cqca}
        \sca=\left(\begin{array}{cc}
          0&1\\
          1&u^{-1}+1+u
        \end{array}
        \right)\in\joliM_2(\Z_2[u,u^{-1}]).
    \end{equation}
    We will use this example throughout the paper. It generates the time evolution depicted in Figure \ref{fig:sca_fermat_sca}.
    \begin{figure}[htbp]
      \centering
      \subfloat[Time evolution of $\sca$ starting from $\xi=\tbinom{1}{0}$.]{\label{fig:sca_fermat_sca}\includegraphics[width=\textwidth/2-3mm]{theta_p2.pdf}}
      \hskip 3mm
      \subfloat[Time evolution of a general nearest neighbor $p$-Fermat CA.]{\label{fig:sca_fermat_fermat}\resizebox{\textwidth/2-3mm}{!}{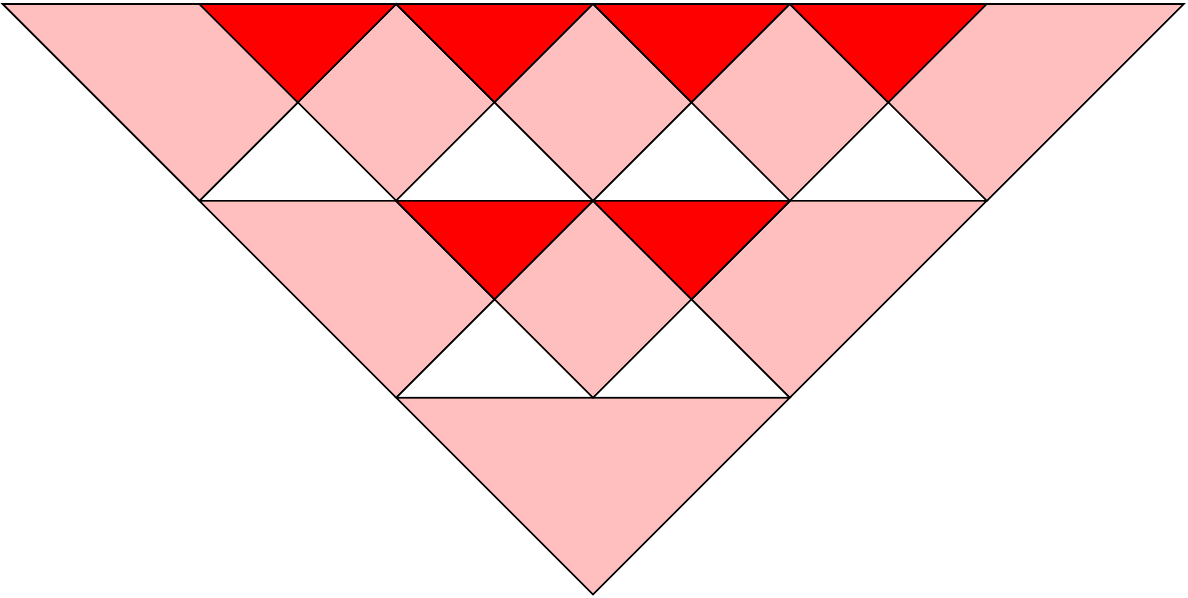}}
      \caption{This figure shows that $\sca$ cannot be a $p$-Fermat CA. In a $p$-Fermat CA at least the white areas are filled by the neutral element $e$; $\sca$ has a different pattern.}
    \end{figure}
    A general nearest-neighbor $p$-Fermat CA produces a time evolution that reproduces itself after $p$ steps in at most three copies located at positions $\acco{-p;0;p}$. After $2p$ steps we have five copies at most.  This creates areas filled with the neutral element $e$ shared by all $p$-Fermat CAs for a fixed $p$.  In figures \ref{fig:sca_fermat_sca} and \ref{fig:sca_fermat_fermat} we can easily see that $\sca$ does not exibit these areas; therefore it is not $p$-Fermat. Furthermore $p$-Fermat CAs that are not periodic are irreversible, while we also allow reversible CAs, $\Theta$ being again one example. 

    \item [\cite{Allouche1996,Allouche1997}] Allouche, Haeseler, Peitgen and Skordev study recurrences in the spacetime diagram of linear cellular automata, from the angle of $k$-automatic sequences, which we will not define in this paper.  However they require $\Sigma$ to be an Abelian ring and the CA to be a ring homomorphism, which is again essentially the case $d=1$.

    \item [\cite{Moore1997,Moore1998}] Moore studies CAs with an alphabet $A$ on a staggered spacetime, where every cell $c$ is only influenced by two cells $a$ and $b$ of the last time step. The update rule is $c=a\bullet b$.  He requires $(A,\bullet)$ to be a quasigroup and studies different special cases. First let us note that these CAs are either irreversible or trivial, while ours don't have to be. Thus, although it is possible to bring our CAs in the form of staggered CAs, the results of Moore do not apply.  In his setting, our CAs would be of the form $c=a\bullet b=f(a)+g(b)$ for some homomorphisms $f$ and $g$. For $(A,\bullet)$ to be a quasigroup means

\begin{equation}
\forall a,b\in A\;\exists! x,y\in A\quad a\bullet x=b\wedge y\bullet a=b.
\end{equation}

In our case, these equalities translate respectively as $g(x)=b-f(a)$ and $f(y)=b-g(a)$. The right-hand sides can be arbitrary elements of $A$, therefore $f$ and $g$ have to be isomorphisms, as indeed required in~\cite{Moore1997}.
    
The angle of study of Moore is also different: he does not exactly study the fractal properties of the spacetime diagram, but rather the complexity of the prediction --- ``What will be the state of this cell after $t$ steps?''.  Describing the spacetime diagram with a matrix substitution system is one way of proving that prediction is an easy task --- for instance it makes it {\bf NC}.

    \item [\cite{macfarlane}] Macfarlane uses Willson's approach and generalizes parts of it to some examples of matrix-valued CAs, including $\sca$.  However, the transition matrix is obtained heuristically --- ``by scrutiny of figure~9'' --- from the spacetime diagram, instead of being algorithmically derived from the transition rule (as in the present work).  
The conclusion (section 6) suggests that the analysis of $\Theta$ is easily generalizable to matrices of various sizes over various rings, so in a sense the present article is but an elaboration of the concluding remark of~\cite{macfarlane}, although we have to say we do not find this generalization to be that obvious. 
\end{description}

The heart of our proof is in section~\ref{section_recursion}.  In a nutshell, whereas most of the techniques used in our article can be traced back to older articles, the new one that allows us to extend the analysis to a larger class of automata is the introduction of $\alpha$ in Equation~(\ref{penultimate_decomposition}).  The idea in doing so is to get rid of the complicated noncommutative ring structure and go back to a simple linear recurrence, as stated in Proposition~\ref{introduce_Xi}.  Since a linear recurrence is precisely where the analysis started from, it could seem at first sight that nothing is gained in the process, but the new recurrence actually does not define a cellular automaton.  Instead of defining line $n+1$ from line $n$, it cuts right through to line $mn$, thus establishing a scaling property.

\subsection{Different CAs}
\label{section_different_ca}
  While we use $\sca$, which has very special properties (being reversible, over a field of characteristic two, and described by a $2\times 2$ matrix), as our example throughout the paper, the analysis applies of course to all other linear CAs. In this section we give a short overview over the variety of spacetime diagrams these CAs generate. Let us start with small changes to $\sca$. For our first example we keep $m=k=2$, but change the determinant to $u$. We only change one entry of the matrix:
  \begin{equation}
    \label{eq:non_cqca}
        \sca_{u}=\left(\begin{array}{cc}
          0&u\\
          1&u^{-1}+1+u
        \end{array}
        \right).
  \end{equation}

  The spacetime diagram is displayed in figure~\ref{fig:non_clifford} and shows how much difference a small change in the update rule can make for the spacetime diagram.

  \begin{figure}[htbp]
    \centering
    \subfloat[Time evolution of $\sca_{u}$.]{\label{fig:non_clifford}\includegraphics[width=\textwidth/2-3mm]{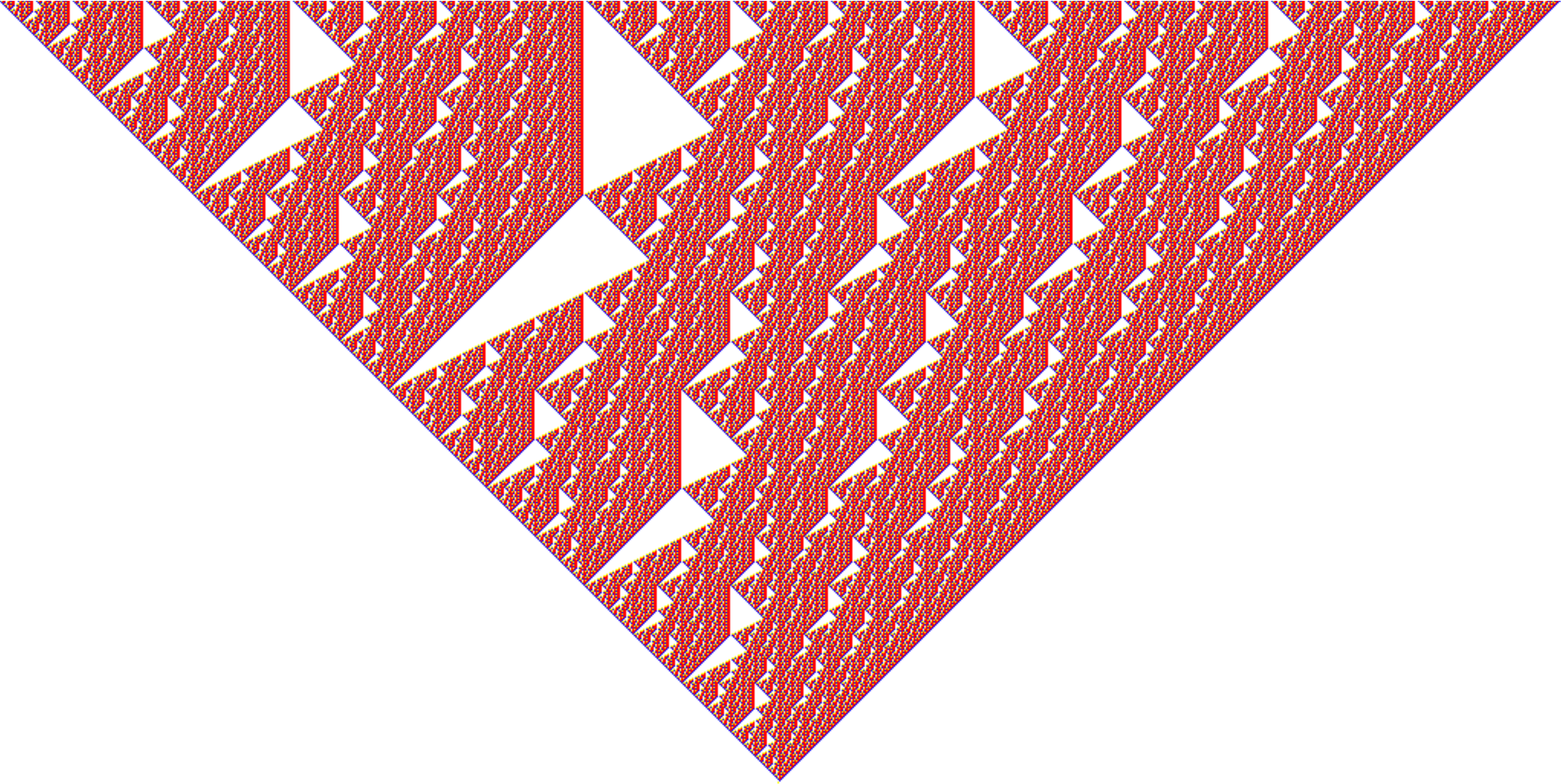}}
    \hskip 3mm 
    \subfloat[Time evolution of $T_{\F_4}$.]{\label{fig:tf4}\includegraphics[width=\textwidth/2-3mm]{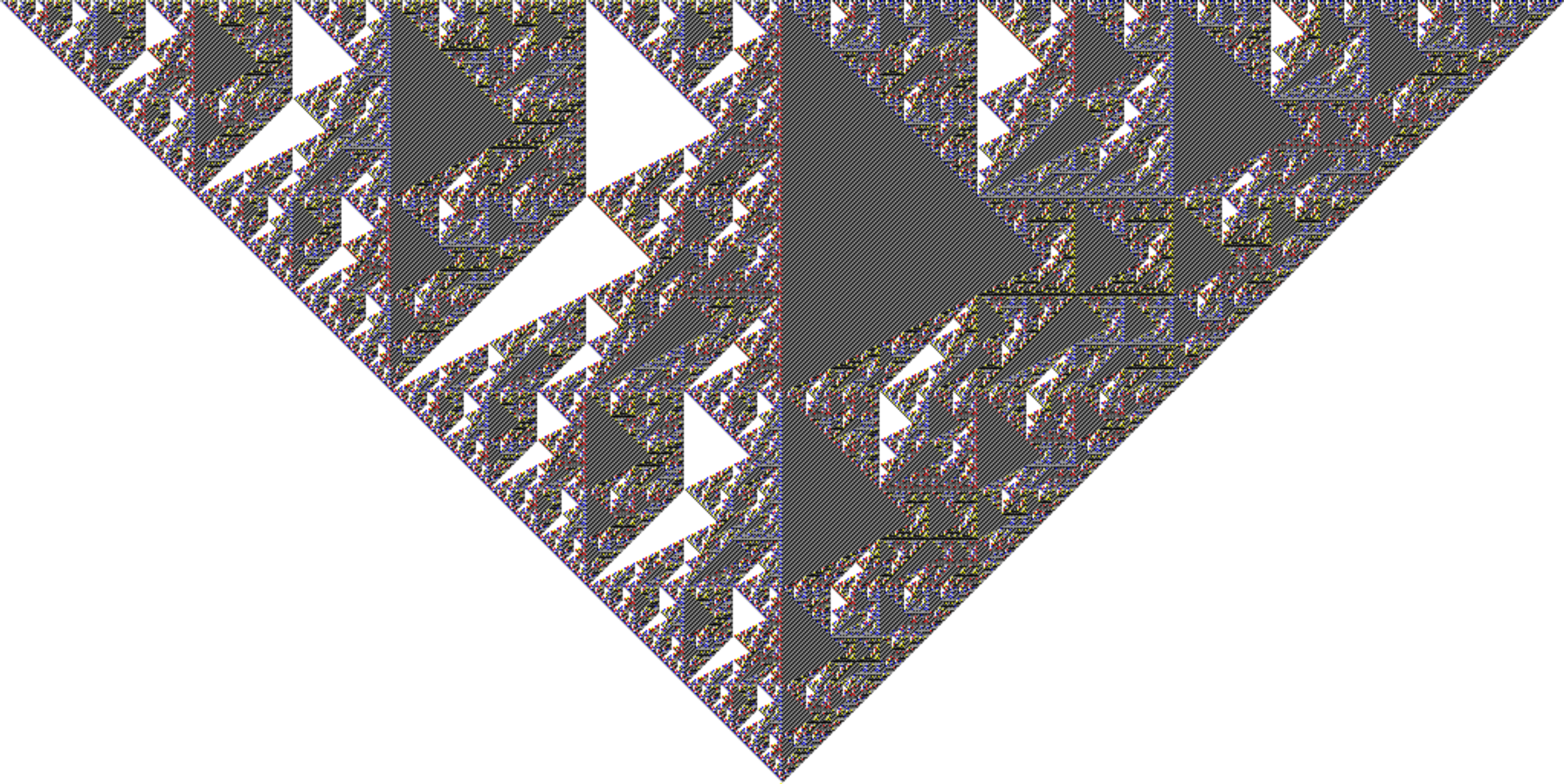}}
    \caption{Spacetime diagrams of non-clifford CAs.}
  \end{figure}

  
Let us now modify $\Theta$ in a more subtle fashion:

  \begin{equation}
    \label{eq:k=4}
        \sca_{k=4}=\left(\begin{array}{cc}
          0&1\\
          1&u^{-1}+1+u
        \end{array}
        \right).
  \end{equation}

The hidden difference with $\sca$ is the underlying ring, which has now been extended from $\Z_2$ to $\Z_4$.  $\sca_{k=4}$ contains in some sense more information than $\Theta$, since $\Theta$ is induced from $\sca_{k=4}$ by the projection $\Z_4\twoheadrightarrow\Z_2$.  Consequently, 
the spacetime diagram of $\sca$ is nothing but a projection of that of $\sca_{k=4}$, as illustrated in figure~\ref{fig:theta_p4}.

  The last CA we want to present lies in $\joliM_2\pa{\F_4[u,u^{-1}]}$, where $\F_4$ is the finite field of order $4$, here identified with $\F_2[\omega]/(\omega^2+\omega+1)$.
  The corresponding matrix is
  \begin{equation}
    \label{eq:f4_ca}
    T_{\F_4}=\left(\begin{array}{cc}
          0&\omega\\
          u^{-1}&(\omega +1)u^{-1}+\omega+u
    \end{array}
    \right).
  \end{equation}
  If one wants to avoid calculations in $\F_4$, this CA can be translated to a CA in $\joliM_4\pa{\Z_2[u,u^{-1}]}$, namely
  $$
    \tilde T_{\F_4}={\left(\begin{array}{cccc} 
          0&0&0&1\\
          0&0&1&1\\
          u^{-1}&0&u^{-1}+u&u^{-1}+1\\
          0&u^{-1}&u^{-1}+1&1+u
    \end{array}
    \right)}.
$$

  Its spacetime diagram, as can be seen shown in figure~\ref{fig:tf4},
contains patches of checkerboard pattern. Somehow, they trivialise most of the usual properties of the figure: for instance they make its fractal dimension $2$, even if the fractal structure can hardly be considered trivial.  In order to access more interesting properties, 
it is possible to blank this pattern out, considering it as just ``another shade of white''.  This can be trivially done on the matrix substitution system, by removing the states from which the blank state is inaccessible.

\subsection{Coloured spacetime diagrams} \label{section_diagram}
The mainstream setting when studying the fractal structure of spacetime diagrams is monochromatic; we introduce colors in the picture.

Instead of considering simple compact subsets of the plane, we will have a finite set of colors $\joliC$ and compact subsets of $\pa{\R^2}^\joliC$.  Let $b\not\in\joliC$ be the additional ``blank'' color and $\mochec:\Sigma\to \joliC\cup\acco{b}$ a coloring of $\Sigma$ such that $\mochec(0)=b$.  To determine a \emph{colored spacetime diagram}, we need furthermore to be given an automaton $T\in\joliM_d\pa{R\cro{u,u^{-1}}}$, an initial state $\xi\in R^d$, and an integer $n$.   The corresponding colored spacetime diagram is then the rescaled diagram obtained by iteratively applying $T$ $n$ times on $\xi$.

Formally, for $n,i,j\in\N$, let $S_{n,i,j}$ be the full square centred in $\frac{1}{n}\pa{i,j}$ and whose edges, parallel to the axes, are of length $\frac{1}{n}$.
To each positive integer $n$ and color $c\in\joliC$ is associated a compact subset of the plane $\joliP_n(c)$ which is the union of the $S_{n,i,j}$'s such that $0\leq j \leq n$ and $\mochec\pa{T^j\pa{\xi}_i}=c$.
The colored spacetime diagram of order $n$ is then the function $\joliP_n:c\mapsto\joliP_n(c)$. A sequence of colored patterns $\pa{\joliP_n}_{n\in\N}$ of spacetime diagrams is said to converge to some colored pattern $\joliP_\infty$ if for every $c\in\joliC$, $\pa{\joliP_n(c)}_{n\in\N}$ converges to $\joliP_\infty(c)$ for the Hausdorff distance.

We can now state our main result.

\begin{thm}\label{bigthm}
Let $G$ be a finite abelian $p$-group.  For every cellular automaton over $G$ that is also a group homomorphism, there exists a positive integer $m$ such that for every fixed initial state the coloured spacetime diagrams of order ${p}^{mn}$ converge when $n$ goes to infinity.
\end{thm}

In general, to know about the fractal structure of a cellular automaton over some finite group $G$, write $G$ as a product of $p$-groups and study each $p$-component of the spacetime diagram independently; according to Theorem~\ref{bigthm}, each component generates a fractal pattern. Then, since the logarithms of the prime numbers are rationally independent, it is possible to find a sequence of resized spacetime diagrams that converges towards a superposition of these different components with arbitrary independent rescaling coefficients, but there is no direct generalization of the theorem.  For instance, even in the simple case of Pascal's triangle modulo 6, there is no real number $\alpha>0$ such that the diagrams of order $\floor{\alpha^n}$ converge; however those of order $t_n$ will converge as soon as the fractional parts of $\log_3(t_n)$ and $\log_2(t_n)$ both converge, and then their limits determine the limit pattern.  The situation is very briefly described in section~5 of \cite{selfsimilarity}.

\subsubsection{Matrix substitution systems}
We will show how to find a suitable description of the limit pattern in the rest of this article.  We now explain exactly what it means to generate a colored picture by rules of substitution, and how to take the limit of all these pictures.  This is a generalization of the usual monochromatic description that can be found for instance in~\cite{mandelbrot,Willson1987,Haeseler1993}, and which corresponds in our setting to the case where all the colors are mapped to ``black''.

Let $V$ be a finite alphabet; because we want colors, compared with the usual definition of a matrix substitution system, we don't have to include a special ``empty'' letter. A matrix substitution system is then a function $\joliD:V \to V^{{\doublecro{1;r}}^2}$; for some integer $r$.  Together with a set of colors $\joliC$ and a coloring $\mochec:V\to\joliC$, it defines colored patterns, much in the same way cellular automata do.  With the previous notations, at each step $n$, the pattern $\joliP_n$ is the union of squares $S_{r^n,i,j}$ of different colors, for different $i$'s and $j$'s; each one of them is indexed by some letter in $V$. 

Then at step $n+1$, each colored square of color $c$ indexed by $v\in V$ present in the $n$th step pattern is replaced by $r^2$ smaller squares that pave it; these smaller squares are given by $\joliD(v)$ and indexed accordingly.  To such a matrix substitution system we can associate a multigraph $\Gamma=(V,E)$  where the set of vertices is $V$ and we put as many edges from $v$ to $w$ as there are $w$'s in $\joliD(v)$.

A \emph{plain} matrix substitution system is one of the usual kind: no coloring, and $V$ contains a special letter $\varepsilon$ such that $\joliD\pa{\varepsilon}$ is a matrix full of $\varepsilon$'s and $\mochec(\varepsilon)=b$.  In the multigraph associated to a plain matrix substitution system, $\varepsilon$ is excluded from the set of vertices.

We want to generalize the usually property of convergence of the patterns defined by plain matrix substitution systems.  This will be done by the conjunction of the two following propositions.  Let us first remind some notions on graphs: the \emph{period} of a graph is the greatest common divisor of the lengths of all the cycles in $\Gamma$; a graph is \emph{aperiodic} if it has period 1. 

\begin{prop}
\label{prop_aperiodic}
If every strongly connected component of $\Gamma$ is aperiodic, then $\pa{\joliP_n}_{n\in\N}$ converges.
\end{prop}

\begin{proof}
To each color $c\in\joliC$ we associate the plain matrix substitution system $\joliD^c$, obtained from $\joliD$ simply by turning some letters into $\varepsilon$. For $v\in V$, let $X_{c}(v)$ be the set of integers $n$ such that there exists a path of length $n$ in $\Gamma$ connecting $v$ to a letter of the color $c$. Since the strongly connected component containing $v$ is aperiodic, $X_{c}(v)$ is either finite or cofinite. Those letters $v\in V$ such that $X_{c}(v)$ is finite are sent to $\varepsilon$, and this defines $\joliD^c$.  If $X_{c}(v)$ is finite and $v'$ can be reached from $v$, then $X_{c}(v')$ is also finite;  therefore, $\joliD^c$ is indeed a substitution system. Let $M$ be such that for every $v\in V$, either $X_{c}(v)$ or its complement is strictly bounded by $M$.

Let us now compare two sequences of figures. The first one is $\pa{\joliP_n(c)}$, the subpattern of color $c$ defined by $\joliD$. The second one is $\pa{\joliP^c_n}$, the one obtained from $\joliD^c$; we know that it converges to some compact $\joliD^c_{\infty}$.
By construction, $\joliP_{n+M}(c)$ is included in $\joliP^c_n$, and for every black square of $\joliP^c_n$, there is a black subsquare in $\joliP_{n+M}(c)$. The Hausdorff distance between $\joliP_{n+M}(c)$ and $\joliP^c_n$ therefore converges to 0, so that $\pa{\joliP_n(c)}$ converges to $\joliD^c_{\infty}$.
\end{proof}

For a graph $\Gamma$, let $\Gamma^k=(V,E^k)$ where $E^k$ is the set of couples $(v,w)$ such that there exists in $\Gamma$ a path of length $k$ from $v$ to $w$.

\begin{prop}\label{prop_existsk}
For every (multi)graph $\Gamma$, there exists $k$ such that every strongly connected component of $\Gamma^k$ is aperiodic.
\end{prop}

\begin{proof}
Each strongly connected component $\Delta$ of $\Gamma$ has a period $p(\Delta)$, so that $\Delta^{p\pa{\Delta}}$ is aperiodic. Let $k_0$ be the least common divisor of the $p(\Delta)$'s; then each strongly connected component of $\Gamma$ induces an aperiodic graph in $\Gamma^{k_0}$, but it is possible that, in the process, it broke down into several connected components, so that $\Gamma^{k_0}$ might not have the required property. The procedure then has to be repeated from $\Gamma^{k_0}$ to obtain $\Gamma^{k_0 k_1}$, and so on. Since the strongly connected components of $\Gamma^{k_0 \cdots k_{i+1}}$ are included in those of $\Gamma^{k_0 \cdots k_{i}}$, this process reaches a fixed point, which is a graph with the required property.
\end{proof}

Ergo, a colored matrix substitution system defines a convergent colored pattern when considering the steps that are a multiple of some well-chosen integer $m$.  So, in order to prove Theorem~\ref{bigthm}, all we need to do is find such a substitution system.  This will be done
in a special case in the next section, and in the general case in section~\ref{proof_sec}.

\section{A special recursion scheme for $\sca$}\label{section_special_scheme}
The aim of this section is to give the most direct and natural explanation of the fractal structure generated by $\sca$ that we are aware of.  Modulo some caveat, it applies effortlessly to all invertible elements $T$ of $\joliM_2\pa{R[u,u^{-1}]}$, where $R$ is a finite Abelian ring of characteristic $2$.  This section is not vital to the proof of the general case presented in~\ref{proof_sec}, and can therefore be skipped by the impatient reader.

We will deduce informally the basic structure of the spacetime diagrams from a simple recursion relation for the $2^n$th powers of $T$. The characteristic polynomial of $T$, $P_T(X)$, is equal to $X^2 + \pa{\tr T} + \det T$. According to Cayley-Hamilton theorem, $P_T(T)=0$, so $T^2 + \pa{\tr T} T + \pa{\det T}\I =0$. Multiplying this equation by $T^{-1}$, we get $T=\pa{\det T}T^{-1}+\pa{\tr T}\I$. Let us denote $\widetilde{T}=\pa{\det T} T^{-1}$, which we will name the \emph{dual} of $T$; since we are in characteristic $2$, by repeatedly taking the square of this equality, we obtain
\begin{equation}
\label{eq:t_to_2_to_n}
\forall n\in\N\quad T^{2^n}=\widetilde{T}^{2^n}+\pa{\tr T}^{2^n}\I. 
\end{equation}

Taking the trace of this equation, we get $\tr T^{2^n}=\tr \widetilde{T}^{2^n}$; in particular, $\tr T=\tr \widetilde{T}$ so Equation~(\ref{eq:t_to_2_to_n}) is also valid when swapping $T$ and $\widetilde{T}$. Let $I_T$ be a finite set, and the $\lambda_i$'s elements of $R$ such that $\tr T = \sum_{i\in I_T}\lambda_i u^i$. Then we have
\begin{equation}
\label{eq:trace_to_2_to_n}
\forall n\in\N\quad \pa{\tr T}^{2^n} = \sum_{i\in I_T}\pa{\lambda_i}^{2^n}u^{2^n i}.
\end{equation}

We do not yet specify the initial state; as a matter of fact, it will prove to be largely irrelevant. The only thing we ask for now is that it is nontrivial (and finite).

Consider for instance $\sca$; we have $\det \sca=1$ and $\lambda_{i}=\chi_{\acco{-1;0;1}}(i)$. We start with the spacetime diagram corresponding to $2^n$ steps; it is rescaled to a triangle with vertex coordinates $\{(0,0),(-1,1),(1,1)\}$. Taking equations (\ref{eq:t_to_2_to_n}) and (\ref{eq:trace_to_2_to_n}), we can see that the state at the $2^n$th time step can be decomposed into a sum of several copies of the initial state (the positions are governed by the coefficients of the trace) and a configuration that can be derived by applying $\widetilde{T}^{2^n}$ to the initial state. In the next $2^n$ steps, this configuration will contract itself to the initial state, which is shifted according to $\log_u \det T$, as $\widetilde{T}$ is the inverse of $T$ composed with the shift $(\det T)\I$. The copies of the original initial state evolve according to $T$. This is illustrated in Figure \ref{fig:2_to_n}.
\begin{figure}[htbp]
  \centering
  \scalebox{0.35}{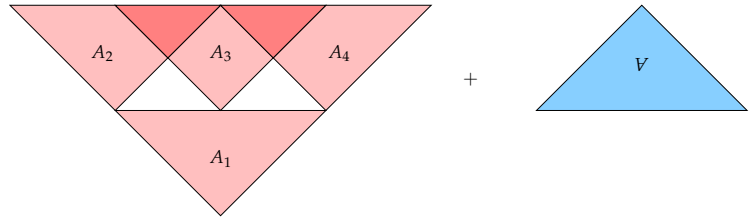}
  \caption{The whole figure is the sum of $\abs{I}+2$ parts.}
  \label{fig:2_to_n}
\end{figure}
The figure suggests to divide the spacetime diagram into four parts $A$, $B$, $C$, and $D$ as shown in Figure~\ref{regleA} which overlap only on a single cell strip at the borders.
\begin{figure}
  \centering
  \subfloat[The first substitution rule.]{\label{regleA}\resizebox{\textwidth/2-3mm}{!}{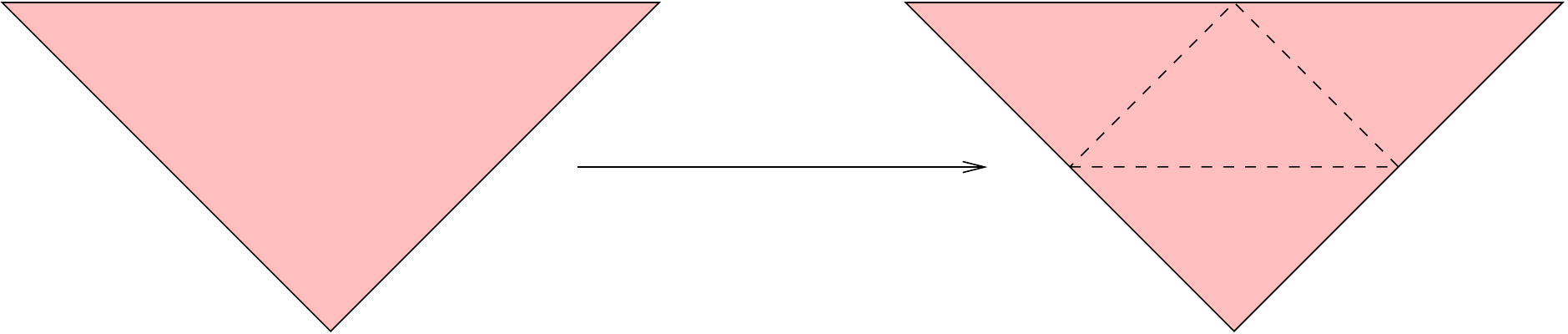}}
  \hskip 3mm 
  \subfloat[The second substitution rule.]{\label{regleB}\resizebox{\textwidth/2-3mm}{!}{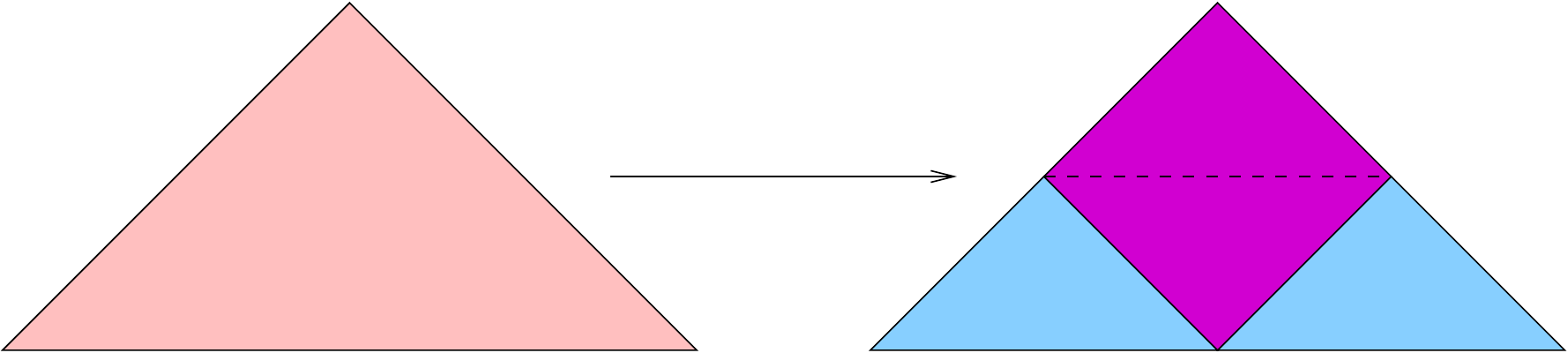}}
  \caption{The first and the second substitution rules.}
\end{figure}

$A_2$, $A_3$, ad $A_4$ are copies of $A_1$, and $\ud{A}$ is marked by an upside down $A$ because it is the reverse evolution of the initial state under the CA $\widetilde{\sca}$ --- so actually it should logically be named $\widetilde{\ud{A}}$ or $\ud{\widetilde{A}}$, but since it always appears upside-down while $A$ always appears straight on its feet, there is no risk of confusion.

Let us assume that the sequence of rescaled spacetime diagrams up to step $2^n$ actually converges. Then that means $A_1$ should be, in the limit, a copy of the whole picture $A$, downsized by a factor $2$, so we rename it $A$. This gives us the first substitution rule, represented in Figure~\ref{regleA}.
The other three parts are still unknown, and we will name these patterns $B$, $C$ and $D$. Equation~\ref{eq:t_to_2_to_n} tells us what the other substitution rules are. 

Since Equation~\ref{eq:t_to_2_to_n} remains true after swapping $T$ and $\widetilde{T}$, $\ud{A}$ admits likewise a partition into $\ud{A}$, $\ud{B}$, $\ud{C}$ and $\ud{D}$. Summing all the parts shown in Figure~\ref{fig:2_to_n}, we get the top rightmost pattern of Figure~\ref{fig:three_steps}.
\begin{figure}[htbp]
  \centering
  \scalebox{0.5}{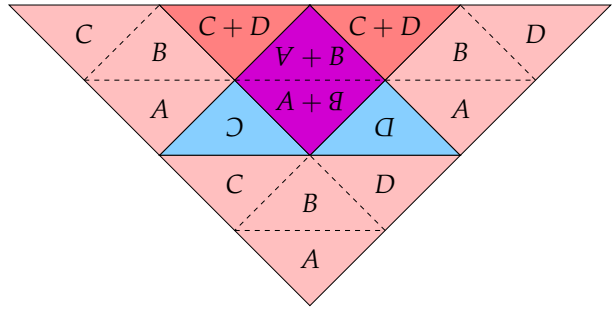}
  \caption{Second step of the decomposition.}
  \label{fig:step2}
\end{figure}
Superimpo-sing our first substitution rule (Figure~\ref{regleA}) with our second step of the decomposition (Figure~\ref{fig:step2}), we get the new substitution rules represented by Figures~\ref{regleB}, \ref{regleC} and~\ref{regleD}.

\begin{figure}
  \centering
  \subfloat[The third substitution rule.]{\label{regleC}\resizebox{\textwidth/2-3mm}{!}{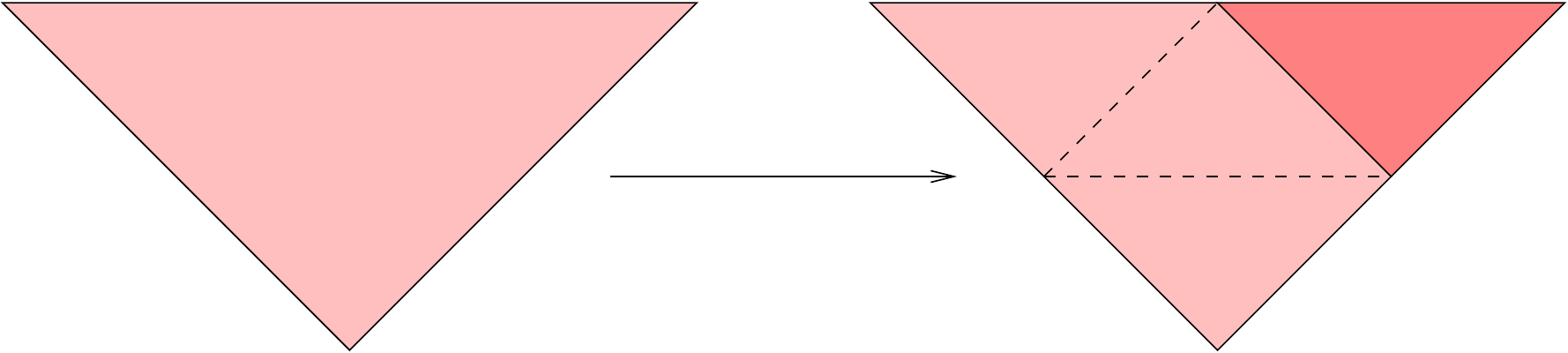}}
  \hskip 3mm 
  \subfloat[The fourth substitution rule.]{\label{regleD}\resizebox{\textwidth/2-3mm}{!}{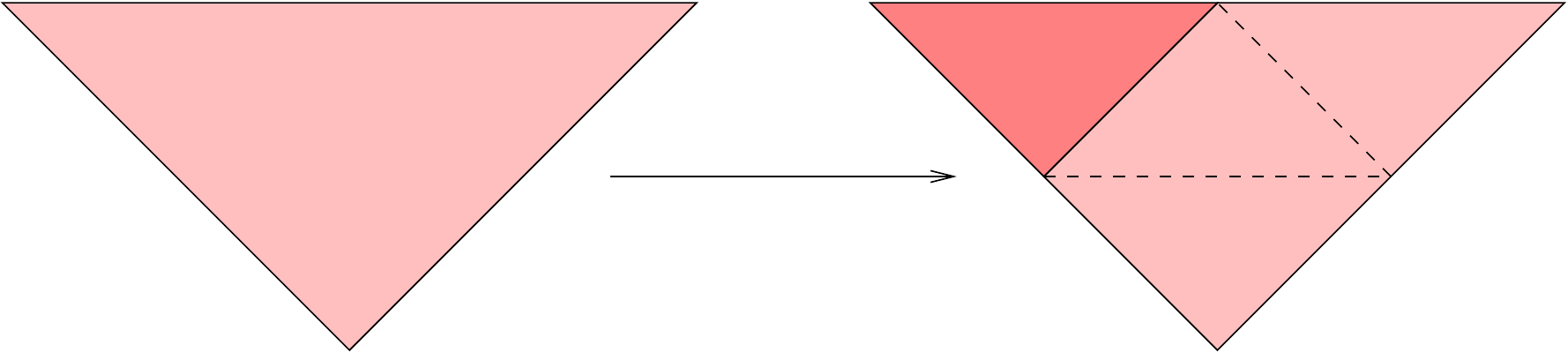}}
  \caption{The third and the fourth substitution rules.}
\end{figure}

All the other substitution rules are deducible from these ones. First they are linear: for instance the substitutions for $C+D$ is the sum of the substitutions for $C$ and $D$, as shown in Figure~\ref{regleC+D}. We don't know \emph{a priori} what the sum of two patterns is, but we know that summing a pattern with itself should give $0$, for we are working in characteristic $2$. In this case $A+A$ and $B+B$ cancel out. In figure~\ref{fig:three_steps}, one can see how the characteristic white spaces emerge from the above substitution rule.  

\begin{figure}[htbp]
  \centering
  \subfloat[The fith substitution rule.]{\label{regleC+D}\resizebox{\textwidth/2-3mm}{!}{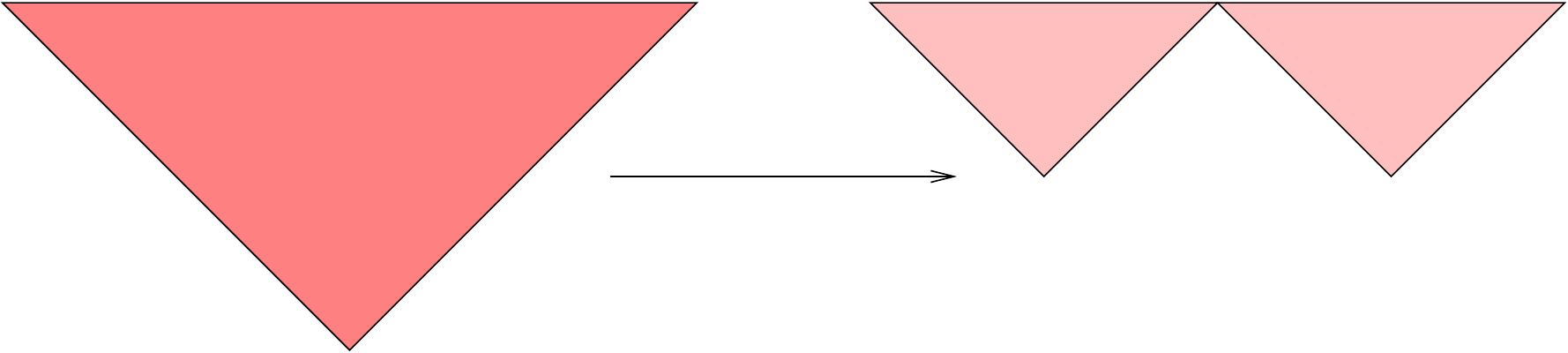}}
  \hskip 3mm 
  \subfloat[The sixth substitution rule.]{\label{regleA+B}\resizebox{\textwidth/2-9mm}{!}{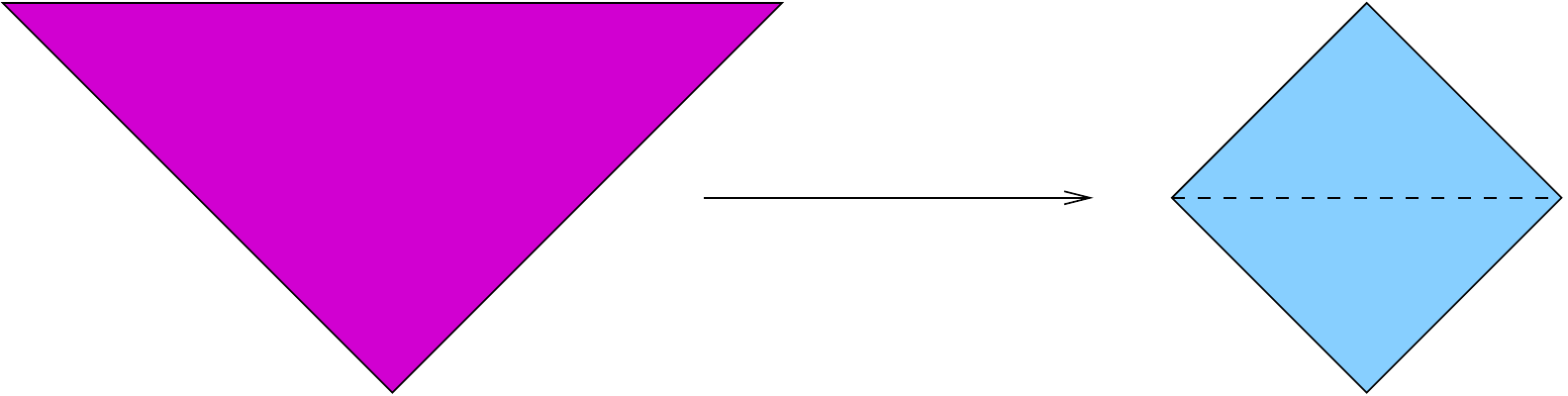}}
  \caption{The fifth and sixth substitution rules are the superpositions of the third and fourth and the first and second rules respectively. One can see, how white spaces emerge because even superpostitions of the same pattern cancel out.}
\end{figure}

The problem with this scheme is that what is happening goes beyond simple juxtaposition of patterns. There can be cancellation at the border between patterns. And, sure, we know $C+C$ is blank, but how do we know $C+D$, for instance, is not? Well, generally we don't. If the initial state is itself blank, then the whole figure would be, and all tiles being blank is certainly a fixed point for all the substitution rules.

In this case, however, everything turns out well. It should first be noticed that in every part of the picture not tagged as ``blank'', an $A$ pattern can be found by refining a few more steps. Formally, let $G=(V,E)$ be the graph whose vertices are the different possible tiles (i.e., in this case, $A$, $B$, $C$, $D$, $\ud{A}$, $\ud{B}$, $\ud{C}$, $\ud{D}$, and the sums modulo $2$ of tiles having compatible shapes, including the blank tile $0$), and edges represent the transition rule in the following way: each vertex has four edges coming out of it, each one pointing to one of its subtiles.
In our case, the graph has the property that the set of vertices accessible from $A$, minus $0$, form a strongly connected component.

We may then distinguish two cases: either $A$ has a point in its interior, or it has points only on its border triangle.  In the first case, the unique non-empty compact defined by the substitution rule is actually the figure we're looking for.  Indeed, it follows from the property of connexity --- cf.\ Proposition~\ref{prop_aperiodic} ---  that every non-zero tile actually appearing in the decomposition of the figure has a non-empty interior. Thus, no matter what happens at the boundary between tiles, the figure constructed this way will always converge to the same compact.

\begin{figure}[htbp]
  \centering
\resizebox{\textwidth}{!}{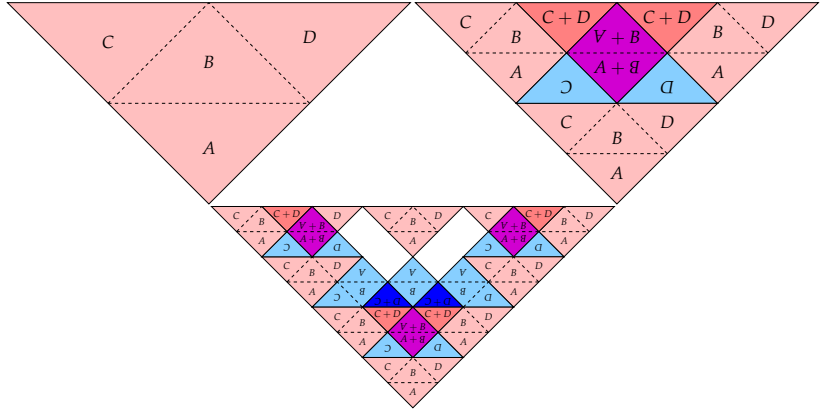}
  \caption{Three decomposition steps for the spacetime diagram of $\sca$. One can clearly see the characteristic white spaces emerging.}
  \label{fig:three_steps}
\end{figure}
\section{Recursion and matrix substitution system}\label{proof_sec}
\label{section_recursion}
We will now present a general method to calculate the fractal dimension and average colour of the spacetime diagrams of linear CAs in $\joliM_d(\Z_{p^l}[u,u^{-1}])$. We will again illustrate the method with our example $\sca$, whereas the derivation is carried out for the general case. Thus the algorithm works as well on all CAs obeying our definition, e.g.\ the CAs presented in section~\ref{section_different_ca},  as it works on $\sca$. Of course with larger neighborhoods and groups of higher order the substitution system becomes larger and larger, so that one might want to use a computer to derive the substitution system.

Our approach is the following: from the minimal polynomial $\Pi$ of the CA $T$ (or any other polynomial fulfilling $\Pi(T)=0$) we derive a recursion relation for the $T^y_x$'s, the coefficients in $u^x$ of $T^y$.  We then forget about every other piece of information we might have on $T$, to concentrate only on this recursion: this shows that the fractal structure, except for contingent blank spaces, can be essentially derived just from the minimal polynomial of $T$.  We further develop our recursion scheme for $T$ until we can express every $T_x^y$ in terms of the $T_i^j$ of the first $m$ time steps with coefficients $\alpha_{j}(x-i,y)$. With a simple grouping of cells we deduce a matrix substitution system that enables us to generate the spacetime diagram of step $t=k^{n+1}$ directly from step $t=k^n$. Using this substitution system we can calculate the fractal dimension, the average colouring, and given an initial state also the whole space time diagram.

Now let $\Pi(X)\in R[u,u^{-1}][X]$ be a monic polynomial such that $\Pi(T)=0$:
\begin{equation}
\Pi(X) = X^m - \sum_{j=0}^{m-1} \lambda_{\Pi,j} X^j.
\end{equation}

According to the Cayley-Hamilton theorem, which we can apply in our case because $T$ is an endomorphism of a finite-dimensional free module over an abelian ring (see Theorem 3.1. of \cite{lang}), the characteristic polynomial of $T$ fulfills this condition, therefore we can always find such a polynomial. 
Let $\polyI$ be the finite set of exponents $i$'s such that the coefficient in $u^i$ of $\Pi$, seen as an element of $R[X][u,u^{-1}]$, is nonzero, so that we can write $\lambda_{\Pi,j}=\sum\limits_{i\in\polyI}\lambda_{\Pi,i,j}u^i$.  $\polyI$ is not to be confused with the neighbourhood of the CA, $I$, which won't play any role from now on.

For any $x,y,n\in\N$, $(x+y)^p\equiv x^p + y^p [p]$, and if $x\equiv y [p^n]$ then $x^p \equiv y^p [p^{n+1}]$.  Therefore, for any $x,y\in R$ and $n\in\N$, $\pa{x+y}^{p^{n+l-1}} = (x^{p^n} + y^{p^n})^{p^{l-1}}$.  Since the powers of $T$ commute pairwise, we get

\begin{equation}
\begin{array}{rcl}
T^{p^{n+2(l-1)} m} & = & \pa{\sum_{j=0}^{m-1} \lambda_{\Pi,j}^{p^{n+l-1}} T^{p^{n+l-1}j}}^{p^{l-1}}\\ &=&  \pa{ \sum_{j=0}^{m-1}\pa{\sum_{i\in \polyI}\lambda_{\Pi,i,j}^{p^{n}}u^{p ^n i}}^{p^{l-1}} T^{p^{n+l-1}j}}^{p ^{l-1}}.\end{array}
\label{eq_brute}
\end{equation}

For each $i,j$, the sequence $(\lambda_{\Pi,i,j}^{p^{n}})$ is ultimately periodic.  There exist therefore integers $N$ and $M$ such that for all $i,j$, $\lambda_{\Pi,i,j}^{p^{M+N}}=\lambda_{\Pi,i,j}^{p^{M}}$.  Let $k=p^N$; substituting $n$ by $M+Nn$ in~(\ref{eq_brute}), we get

\begin{equation}
T^{k^n p^{M+2(l-1)} m } =  \pa{ \sum_{j=0}^{m-1}\pa{\sum_{i\in \polyI}\lambda_{\Pi,i,j}^{p^M}u^{k^np^M i}}^{p^{l-1}} T^{k^{n}p^{M+l-1}j}}^{p ^{l-1}}.
\end{equation}

Hence, if we note $m'= p^{M+2(l-1)} m$ and expand this equation, we find that there is some finite subset $\polyI'$ of $\Z$ and some elements $\mu_{i,j}$ of $R$, for $i\in\polyI'$ and $j\in\doublecro{0;m'-1}$, such that for all $n\in\N$,

\begin{equation}
T^{k^n m'} =  \sum_{j=0}^{m'-1}\sum_{i\in\polyI'}\mu_{i,j} u^{k^n i} T^{k^{n}j}.
\end{equation}

We have now used everything we needed to know from the multiplicative structure on the ring of matrices.  As announced at the end of section~\ref{sec_related_work}, we will now get rid of it and concentrate only on the linear recurrence relation that we have just derived.  Remember that $T^j\in \joliM_d\pa{R\cro{u,u^{-1}}}$, and we are interested in the coefficient of $T^j$ in $u_i$, denoted $T^j_i$, so that $T^j=\sum\limits_{i\in\polyI} T^j_i u^i$.  
We thus get the following relation: $T^{k^n m'+ y}_x = \sum_{i\in \polyI'}\sum_{j=0}^{m'-1}  \mu_{i,j} T_{x-k^n i}^{y+k^n j}$, which we rewrite in this form:

\begin{equation}
T^{y}_x = \sum_{(i,j)\in \polyI'\times\doublecro{0;m'-1}}  \mu_{i,j} T^{g_{i,j}(y)}_{x+f_{i,j}(y)}
\label{barbibulle}
\end{equation}

where $f_{i,j}(y)= -k^n i$ and $g_{i,j}(y)=y- k^n (m'-j)$, which of course works with any $n$, but we will choose $n=\floor{\log_k \frac{y}{m'}}$.  In order to emphasise that the rest of the proof will use only a minimal structure, we state in the next proposition what will actually be proven, and change the notation from $T$, which was an element of $\joliM_d\pa{R[u,u^{-1}]}$, to $\Xi$, an element of a more arbitrary $R$-module.  It is straightforward to check that $T$ fulfills the hypotheses of the proposition.

\begin{prop}\label{introduce_Xi}
Let $M$ be a finite $R$-module, $k$ a positive integer, $\Lambda$ a finite set of indices, and for $i\in\Lambda$, $\mu_i\in R$, $f_i:\llbracket m ; + \infty \llbracket\to\Z$ and $g_i:\llbracket m ; + \infty \llbracket\to\N$ such that for all $y\in \llbracket m ; + \infty \llbracket$ and $t\in\doublecro{0;k-1}$, 

\begin{itemize}
\item $g_i(y)<y$; 

\item $f_i(ky+t)=k f_i(y)$ and $g_i(ky+t)=k g_i(y)+t$.

\end{itemize}

For $x\in\Z\times\N$, let $\Xi^y_x\in M$ be such that when $ y \geq  m$,	

\begin{equation}
\Xi^y_x =\sum_{i\in\Lambda}\mu_i \Xi^{g_i(y)}_{x+f_i(y)}.
\label{eqX}
\end{equation}

Then there exists a finite set $E$ and a function $e:\Z\times\N\to E$ such that

\begin{itemize}
\item $\Xi^y_x$ is a function of $e(x,y)$;
\item for $s,t\in\doublecro{0;k-1}$, $e\pa{kx+s,ky+t}$ is a function of $s$, $t$, and $e(x,y)$.
\end{itemize}

\end{prop}

The introduction of a new function $e$ in this proposition comes from the need of a scaling property, expressing that the state at point $(kx+s,ky+t)$ can be deduced from the state at point $(x,y)$.  Such a property does not follow immediately from Equation~(\ref{eqX}), but it is possible to expand the state space from $M$ to $E$, and to put more information into $e$ than into $\Xi$, so as to fulfill the scaling property.  An immediate consequence of this proposition is that the spacetime diagrams of $\Xi^y_x$ of order $k^n$ can be described by colored matrix substitution systems, so that Theorem~\ref{bigthm} will follow from Proposition~\ref{prop_existsk}.  Let us now prove Proposition~\ref{introduce_Xi}.

\begin{figure}[htbp]
\centering
\scalebox{0.65}
{\Tree [.(1,3) [.(0,2) (-1,1) (0,1) (1,1) (0,0) ]   [.(1,2) (0,1) (1,1) (2,1) (1,0) ] [.(2,2) (1,1) (2,1) (3,1) (2,0) ] (1,1) ]}
\caption{Recursive calls to (\ref{eqX}) for $T=\sca$\label{arbre1}.}
\end{figure}

If $y\geq m$, we can recursively apply Equation~(\ref{eqX}) to give an expression of $\Xi_x^y$ in terms of a linear combination of $\Xi_{x'}^{y'}$'s with $y'< y$.  

Let us consider the case $T=\sca$ as an example. In this case, $R=\F_2$, $k=m=2$, and $\Pi_\Theta(X)=X^2 + \pa{u^{-1}+1+u}X + 1$.  Equation (\ref{eqX}) gives $\Xi^3_1=\Xi^1_1 + \Xi^2_{0}+\Xi^2_1+\Xi^2_2$.  Whereas the decomposition halts here for $\Xi_1^1$, it goes on for the other three terms of the sum, again with $n=0$. For instance, $\Xi^2_{0}$ is transformed into $\Xi^0_{0}+\Xi^1_{-1}+\Xi^1_{0}+\Xi^1_{1}$.  This gives the tree of recursive calls in Figure~\ref{arbre1}, where each node is labelled $(x,y)$ for $\Xi_x^y$; in that particular case, because $R=\F_2$, the only possible coefficient are $0$ and $1$, so they are simply represented by the absence or the presence of the corresponding term.

Starting from any point $(x,y)\in\Z\times\N$ and using recursively equation~(\ref{eqX}), we get to the expression 
\begin{equation}
\Xi_x^y = \sum_{i}\sum_{j=0}^{m-1}\alpha_{\Pi,i,j}(x,y) \Xi_i^j,
\label{penultimate_decomposition}
\end{equation}
which we take as a definition of $\alpha_{\Pi,i,j}(x,y)$.  On Figure~\ref{arbre1}, this corresponds to (parity) counting the leaves tagged $(i,j)$; for instance, there are 4 leaves tagged $(1,1)$, so $\alpha_{\Pi_\sca,1,1}(1,3)=0$ and $(1,1)$ does not appear as a node in Figure~\ref{arbre1}.

\begin{figure}[htbp]
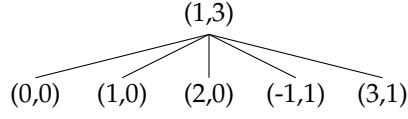

\centering
\Tree [.(1,3) (0,0) (1,0) (2,0) (-1,1) (3,1) ]
\caption{\label{arbre2}$\Xi_1^3$ as a sum of the simplest terms, when $T=\Theta$.}
\end{figure}

Since the relation~(\ref{eqX}) is invariant under translations of the parameter $x$, we have $\alpha_{i,j}(x,y)=\alpha_{0,j}(x-i,y)$.  Noting $\alpha_{j}:=\alpha_{0,j}$, we get the following equation:
\begin{equation}
\Xi_x^y = \sum_{i\in\Z}\sum_{j=0}^{m-1}\alpha_{j}(x-i,y) \Xi_i^j.
\label{ultimate_decomposition}
\end{equation}

Let us now show that $\alpha_{\cdot}\pa{kx+s,ky+t}$, for $s,t\in\doublecro{0;k-1}$, is a function of $s$, $t$, and the $\alpha_{\cdot}\pa{x', y}$'s, where $x'$ ranges over some neighbourhood of $x$.  By substituting $x$ with $kx+s$ and $y$ with $ky+t$ in Equation~(\ref{eqX}), we get
\begin{equation}
\label{eq:recursion_1}
\Xi^{ky+t}_{kx+s}=\sum_{i\in \Lambda}\mu_{i}\Xi^{kg_i(y)+t}_{k\pa{x+f_i(y)}+s}.
\end{equation}

Because the indices and exponents of $\Xi$ on the left and right side of this equation have undergone the same transformation $(x,y)\mapsto (kx+s,ky+t)$,  we arrive recursively at this point

\begin{equation}
\Xi^{ky+t}_{kx+s}= \sum_{i\in\Z}\sum_{j=0}^{m-1} \alpha_{j}(x-i,y) \Xi^{kj+t}_{ki+s}
\label{end_of_recursion}
\end{equation}

which we want to compare to the following equation, directly deduced from (\ref{ultimate_decomposition}):
\begin{equation}
\Xi^{ky+t}_{kx+s}= \sum_{i\in\Z}\sum_{j=0}^{m-1} \alpha_{j}(kx+s-i,ky+t) \Xi^{j}_{i}
\label{ultimate_decomposition_kx}
\end{equation}

Of course, there can be terms in (\ref{end_of_recursion}) with $kj+t\geq m$, so that the decomposition is not over: it then needs to be performed to its end.  For instance, in our example $T=\Theta$, if instead of the tree of recursive calls for $\Xi_1^3$ we want that of $\Xi_3^6=\Xi_{2\times 1 + 1}^{2\times 3}$, we know that the upper part of the tree will be Figure~\ref{arbre3}.

\begin{figure}[htbp]
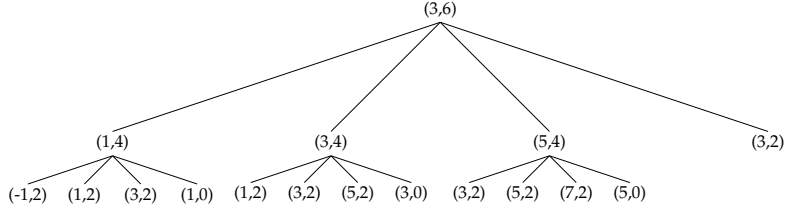

\centering
\scalebox{0.65}
{\Tree [.(3,6) [.(1,4) (-1,2) (1,2) (3,2) (1,0) ]   [.(3,4) (1,2) (3,2) (5,2) (3,0) ] [.(5,4) (3,2) (5,2) (7,2) (5,0) ] (3,2) ]}
\caption{\label{arbre3}Recursive calls to (\ref{eqX}) for $T=\Theta$, minus some leaves.}
\end{figure}

As for the lower part, the decomposition is not over, so it needs to be performed to its end.  Thanks to linearity, we don't need to remember the whole tree, just its leaves.  As is well known, a tree grows from its leaves, and in the same way we need only apply the transformation $(i,j)\mapsto (2i+1,2j)$ to the tree in Figure~\ref{arbre2} in order to find the correct decomposition for $\Xi_3^6$, which gives Figure~\ref{arbre4}.

\begin{figure}[htbp]
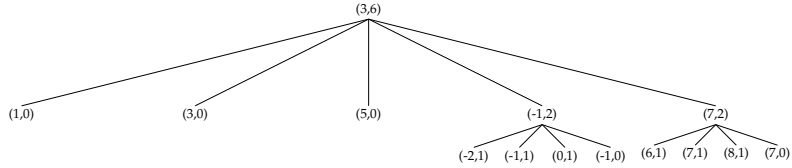

\centering
\scalebox{0.5}
{\Tree [.(3,6) (1,0) (3,0) (5,0) [.(-1,2) (-2,1) (-1,1) (0,1) (-1,0) ] [.(7,2) (6,1) (7,1) (8,1) (7,0) ] ]}
\caption{\label{arbre4}Trees grow from their leaves.}
\end{figure}

However, even keeping track only of the leaves would require an unbounded alphabet, that is why we 
defined $\alpha_{\Pi,j}(x,y)=\alpha_{\Pi,i,j}(x+i,y)$, i.e.\ from the tree for $\Xi_x^y$ we want to keep track only of the leaves labelled by $(0,j)$.  Alas, that will not work so easily: in Figure~\ref{arbre4}, $(0,1)$ appears once as a son of $(-1,2)$, which comes from node $(-1,1)$ in Figure~\ref{arbre2}.  So it appears we cannot just forget about everything else.

As we just saw on one example, $\alpha_{\Pi,\cdot}(kx+s,ky+t)$ typically does not depend on $\alpha_{\Pi,\cdot}\pa{x,y}$.  Instead the final decomposition of (\ref{end_of_recursion}) will relate the coefficients $\alpha_{\Pi,j}(kx+s-i,ky+t)$ of the $\Xi^j_i$ in (\ref{ultimate_decomposition_kx}) to sums of the $\alpha_{\Pi,j}^k(x-i,y)$'s. Therefore $\alpha_{\Pi,\cdot}(kx+s,ky+t)$ depends on  
the $\alpha_{\Pi,\cdot}\pa{x+i,y}$'s for $i$ ranging over some finite set $\joliS$  depending only on $\Pi$. However, this is not much of a problem, as a simple grouping will take care of it --- a technique commonly attributed to~\cite{Willson1987}. Let us show that $\joliS$ is finite. Since $j\in\acco{0;\ldots;m-1}$ and $t\in\acco{0;\ldots;k-1}$, the $kj+t$ appearing as an exponent of $\Xi$ in~(\ref{end_of_recursion}) is in $\acco{0;\ldots;km-1}$. We therefore have to use at most $\pa{k-1}m$ recursive calls to~(\ref{eqX}) in order to get down to coefficients $\Xi_x^y$ with $y<m$, each one of them decreasing the exponent by at least $1$.  Each one of them also increases the index by $f_i(y)$; since both $\Lambda$ and $\doublecro{0;km-1}$ are finite, the set of possible $f_i(y)$'s is also bounded by some $M$, and the total variation in the index, i.e. $\joliS$, is then bounded by $\pa{k-1}mM$; let us say $\joliS\incl \doublecro{d_{\min};d_{\max}}$.

Let us now introduce $\beta_{\Pi,\cdot}(x,y)= \pa{\alpha_{\cdot}(x-i,y)}_{i\in \joliS'}$, where $\joliS'= \doublecro{\delta_{\min};\delta_{\max}}$, $\delta_{\max}$ being such that $\delta_{\max}\geq d_{\max}+\ceil{\frac{\delta_{\max}}{k}}$ and $\delta_{\min}$ such that $\delta_{\min}\leq d_{\min}-1 + \ceil{\frac{\delta_{\min}+1}{k}}$. This time, for $s,t\in\acco{0;\ldots;k-1}$, $\beta_{\Pi,\cdot}\pa{kx+s,ky+t}$ does really depend only on $\beta_{\Pi,\cdot}\pa{x,y}$.  Indeed, 
$$\beta_{\Pi,\cdot}\pa{kx+s,ky+t}=\pa{\alpha_{\cdot}(kx+s-i,ky+t)}_{i\in \joliS'},$$
and each $\alpha_{\cdot}\pa{kx+s-i,ky+t}$ depends on $\pa{\alpha_{\cdot}\pa{x+\floor{\frac{s-i}{k}}-j,y}}_{j\in\joliS}$ only; the choice of $\joliS'$ has been made so that $j-\floor{\frac{s-i}{k}}\in\joliS'$.  This concludes the proof of Proposition~\ref{introduce_Xi}, as we can choose $E=M^{\doublecro{0;m-1}\times\joliS'}$, with $e(x,y) (j,i)= \alpha_{j}(x-i,y)$.

\subsection{Example: $\sca$}
In the case of $\sca$, Equation~(\ref{end_of_recursion}) becomes
\begin{equation}
\Xi^{2y+t}_{2x+s}= \sum_{i} \alpha_{\Theta,0}(x-i,y) \Xi^{t}_{2i+s}+\alpha_{\Theta,1}(x-i,y) \Xi^{2+t}_{2i+s}.
\end{equation}

The first term is now elementary, but the second one has to be decomposed once more, i.e.
\begin{equation}
\Xi^{2+t}_{2i+s}= \Xi^t_{2i+s} + \Xi^{1+t}_{2i+s-1}+\Xi^{1+t}_{2i+s}+\Xi^{1+t}_{2i+s+1},
\end{equation}
which is the end of it if $t=0$, but not if $t=1$, where we get
\begin{equation}
\Xi^3_{2i+s} =  \Xi^0_{2i+s-1} + \Xi^0_{2i+s}+\Xi^0_{2i+s+1}+\Xi^{1}_{2i+s-2}+\Xi^{1}_{2i+s+2}.
\end{equation}
Comparing this with (\ref{ultimate_decomposition_kx}) we can deduce the substitution rule of $\alpha_{\sca,\cdot}$. It can then be written in the following way, where for convenience $\alpha_{\sca,\cdot}$ is represented as $\begin{array}{|c|}\hline\alpha_{\sca,1}\\ \alpha_{\sca,0}\\ \hline \end{array}$:

\begin{displaymath}
\begin{array}{c}
\begin{array}{|c|}
  \hline
  \alpha_{\sca,\cdot}(x,y)\\
  \hline
\end{array}\\
\downarrow\\
\begin{array}{|c|c|}
  \hline
  \alpha_{\sca,\cdot}(2x,2y+1) & \alpha_{\sca,\cdot}(2x+1,2y+1)\\
  \hline
  \alpha_{\sca,\cdot}(2x,2y) & \alpha_{\sca,\cdot}(2x+1,2y)\\
  \hline
\end{array}\\
\rotatebox[origin=c]{90}{=}\\
\scalebox{0.85}
{$\begin{array}{|c|c|}
  \hline
  \begin{array}{c} \alpha_{\sca,0}\pa{x,y}+\alpha_{\sca,1}\pa{x-1,y}+\alpha_{\sca,1}\pa{x+1,y} \\ \alpha_{\sca,1}(x,y) \end{array} & \begin{array}{c} 0\\ \alpha_{\sca,1}(x,y)+\alpha_{\sca,1}(x+1,y) \end{array}\\
  \hline
  \begin{array}{c}\alpha_{\sca,1}(x,y)  \\ \alpha_{\sca,0}\pa{x,y}+\alpha_{\sca,1}\pa{x,y}\end{array} &  \begin{array}{c} \alpha_{\sca,1}(x,y)+\alpha_{\sca,1}(x+1,y)  \\ 0 \end{array}\\
  \hline
\end{array}$}
\end{array}
\end{displaymath}

If we follow exactly what has been said in the general case, we ought to consider the grouping $\doublecro{-2;3}$.  However, this general bound is obviously too rough in the case of $\Theta$, where we will just have to take $\doublecro{-1;2}$.
We will represent the grouping in the form 
\begin{displaymath}
\begin{array}{|cccc|}
  \hline
  \alpha_{\sca,1}(x-1,y) & \alpha_{\sca,1}(x,y) & \alpha_{\sca,1}(x+1,y) & \alpha_{\sca,1}(x+2,y)\\
  \alpha_{\sca,0}(x-1,y) & \alpha_{\sca,0}(x,y) & \alpha_{\sca,0}(x+1,y) & \alpha_{\sca,0}(x+2,y)\\ 
  \hline
\end{array}
\end{displaymath}
The alphabet has thus size 256, and the substitution system is described by

\begin{center}
\scalebox{0.8}
{$\begin{array}{c}
 \begin{array}{|c|}
    \hline
    \begin{array}{cccc}
      a & b & c & d\\
      e & f & g & h
    \end{array}\\
    \hline
  \end{array}
\\
  \downarrow
\\
  \begin{array}{|c|c|}
    \hline
    \begin{array}{cccc}
      0 & a+c+f & 0 & b+d+g \\
      a+b & b & b+c & c 
    \end{array} 
    & 
    \begin{array}{cccc}
      a+c+f & 0 & b+d+g & 0 \\
      b & b+c & c & c+d 
    \end{array}
    \\
    \hline
    \begin{array}{cccc}
      a+b & b & b+c & c\\
      0 & b+f & 0 & c+g
    \end{array} 
    & 
    \begin{array}{cccc}
      b & b+c & c & c+d \\
      b+f & 0 & c+g & 0 
    \end{array}
    \\
    \hline
  \end{array}
\end{array}
$}
\end{center}

Let us denote by $A$ the matrix having a $1$ in position $a$ and $0$ elsewhere, $B$ the matrix having a $1$ only in position $b$, and so on. For these matrices, we will denote the sum of matrices by a simple juxtaposition: $AB$ will mean $A+B$, as the matrix multiplication has no meaning in this context.

Since $T^0_x=\delta_{x,0} T^0_0$, the starting position, with which we describe the whole line number $0$, is $\begin{array}{c|c|c|c|c|c|c|c}\hline \cdots&0&H&G&F&E&0&\cdots\\ \hline\end{array}$. Since we have, for instance, the rule 
$
\begin{array}{|c|}
  \hline
  F\\
  \hline
\end{array}
\to 
\begin{array}{|c|c|}
  \hline
  B & A \\
  \hline
  F & E\\
  \hline
\end{array}
\label{theta_scheme_2}$, the graph derived from this substitution system is aperiodic; that means that, in whatever way $A,B,C,\ldots$ are represented, either as colored dots or as white dots, the pattern converges, and the fractal structure is described by this matrix substitution system (see Section~\ref{section_diagram}).

To calculate the fractal dimension of our spacetime diagram we use the transition matrix of the matrix substitution system, which contains the information about the images of all states. The line corresponding to $F$ would contain a $1$ in the rows of $A$, $B$, $E$, and $F$ and zeros elsewhere. As every cell gives rise to four new cells the sum of all entries in each column of the matrix is $4$. We thus deal with a sparse $256 \times 256$ matrix.  The base $2$ logarithm of the second largest eigenvalue of this matrix is the fractal dimension of the spacetime diagram (cf.\ for instance \cite{Willson1987}).  Here this gives a fractal dimension of $\log_2\frac{3+\sqrt{17}}{2}\simeq 1.8325$, as also found in \cite{macfarlane}. 

Let us note that up to this point our analysis for $\sca$ is word for word valid for all CA in $\joliM_2(\Z_2[u,u^{-1}])$ of determinant $1$ and trace $u^{-1}+1+u$. The additional information is only used for the actual colouring of the picture. In general all CA with the same minimal polynomial have the same substitution system, and in dimension $2$ the minimal polynomial is entirely determined by the trace and the determinant. The fractal we get if we use the substitution system starting from $\begin{array}{c|c|c|c|c|c|c|c}\hline \cdots&0&H&G&F&E&0&\cdots\\ \hline\end{array}$ is shown in Figure~\ref{fig:sub_theta}.

\begin{figure}[htbp]
  \centering
  \includegraphics[width=\textwidth]{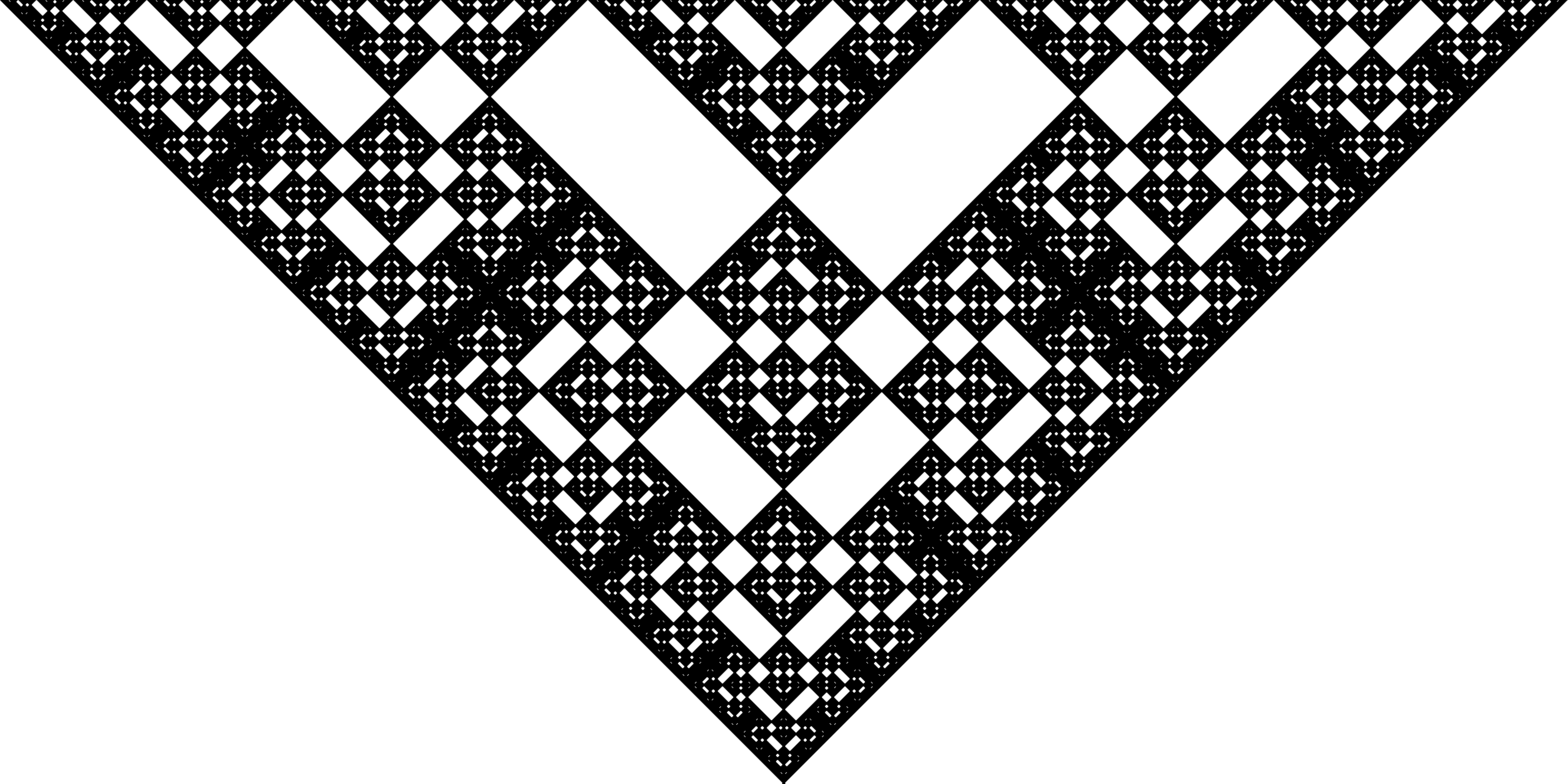}
  \caption{The general fractal time evolution of a linear CA with $k=m=2$, determinant $1$ and trace $u^{-1}+1+u$. Only the areas where the whole group of $\alpha$s is $0$ is marked white. Thus the image appears to have less white than the colored picture. In the limit of infinite recursion this effect vanishes, thus the fractal that is generated is actually the same.}
  \label{fig:sub_theta}
\end{figure}

In the case of $\Theta$ the connection between the substitution system and the colored picture is very simple; let us take $\xi=\tbinom{1}{0}$ as the initial state. Then the state of cell $x$ after $y$ iterations is 

\begin{equation}\begin{array}{rcl}\sca_x^y \xi&=&\sum\limits_i \pa{\alpha_{\sca,0}(x-i,y) \sca_i^0 + \alpha_{\sca,1}(x-i,y) \sca_i^1}\xi\\& = &\alpha_{\sca,0}(x,y) \xi + \sum\limits_i \alpha_{\sca,1} (x-i,y)\sca_i^1\xi.
\end{array}
\end{equation}
 
Since $\sca_0^1=\pa{\begin{array}{cc}0 & 1 \\ 1 & 1 \end{array}}$, $\Theta_1^1=\sca_{-1}^1=\pa{\begin{array}{cc}0 & 0 \\ 0 & 1 \end{array}}$ and $\Theta_i^1=0$ when $i\not\in\acco{-1;0;1}$, we have $T_x^y\xi=\pa{\begin{array}{c}\alpha_{\sca,0} (x,y) \\ \alpha_{\sca,1}(x,y) \end{array}}$. This gives us a color assignment for each state of the matrix substitution system, which corresponds to simply dropping all states that include neither $B$ nor $F$.

We can now  determine the average hue of the spacetime diagram making use of the eigenvector corresponding to the second largest eigenvalue of the transition matrix \cite{Willson1987}. Let us say $\tbinom{1}{0}$, $\tbinom{0}{1}$ and $\tbinom{1}{1}$ are respectively coded by the colors $c_{10}$, $c_{01}$ and $c_{11}$; let $c_q$ be the white color. We determine which symbols of the alphabet belong to each of the colors by looking only at the part $\tbinom{\alpha_{\sca,0} (x,y)}{\alpha_{\sca,1}(x,y)}$. Then we just add up all the weights of symbols with the same color in the eigenvector. We get the following unnormalised coefficients: $c_{10}$: $2(4+\sqrt{17})$, $c_{01}$: $2(4+\sqrt{17})$, and $c_{11}$: $5+\sqrt{17}$. 

In figure~\ref{fig:theta}, this color code was used: $c_{10}=\scalebox{0.5}{\includegraphics{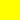}}$, $c_{01}= \scalebox{0.5}{\includegraphics{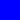}}$ and $c_{11}=
\scalebox{0.5}{\includegraphics{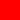}}$. We must therefore have the following average hue: 
\scalebox{0.5}{\includegraphics{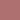}}.

\section*{Conclusion}

We have shown that every cellular automaton inducing a morphism of abelian groups produces a spacetime diagrams that converge to selfsimilar patterns that are superpositions of elementary patterns obtained for groups of prime power order.  We exhibited an algorithm taking as input the local transition rule and outputting a description of these elementary patterns.  
We only studied the one-dimensional case in this article, but the analysis can be carried over to higher dimensions with not much more ado.  Instead of $\joliM_d\pa{R[u,u^{-1}]}$, a $n$-dimensional linear cellular automaton would then be an element of $\joliM_d(R[u_1,u_1^{-1},u_2,u_2^{-1},\ldots u_n,u_n^{-1}])$, matrix substitution systems would become $(n+1)$-dimensional array substitution systems, the system of indices in section $3$ would be further complicated, and the spacetime diagrams would be harder to display.  However, the generalization does not present any theoretical difficulty and makes nice pictures.

For instance, Figure~\ref{fig:fractal3d} is a view of the the fractal structure produced by  $ \Omega=\pa{\begin{array}{cc}0&1\\ 1&(1+u+u^{-1})(1+v+v^{-1})\end{array}}$, a variation on $\Theta$ belonging to $\joliM_2\pa{\Z_2[u,u^{-1},v,v^{-1}]}$.  For obscure historical reasons, pyramids used to be constructed with their points facing up, so to continue the tradition, this fractal is represented upside-down compared to the two-dimensional figures.  Its neighborhood being a square centered on $(0,0)$, the pyramid has unsurprisingly a square base.  We cut it open so as to reveal its insides.  The sector that is cut out is delimited by two half-lines, the one that is visible on the right along $(1,0)$, the sector having an arbitrary angle of approximately $\frac{2}{3}\pi$.

\begin{figure}[htbp]
  \centering
\scalebox{0.155}{
\includegraphics{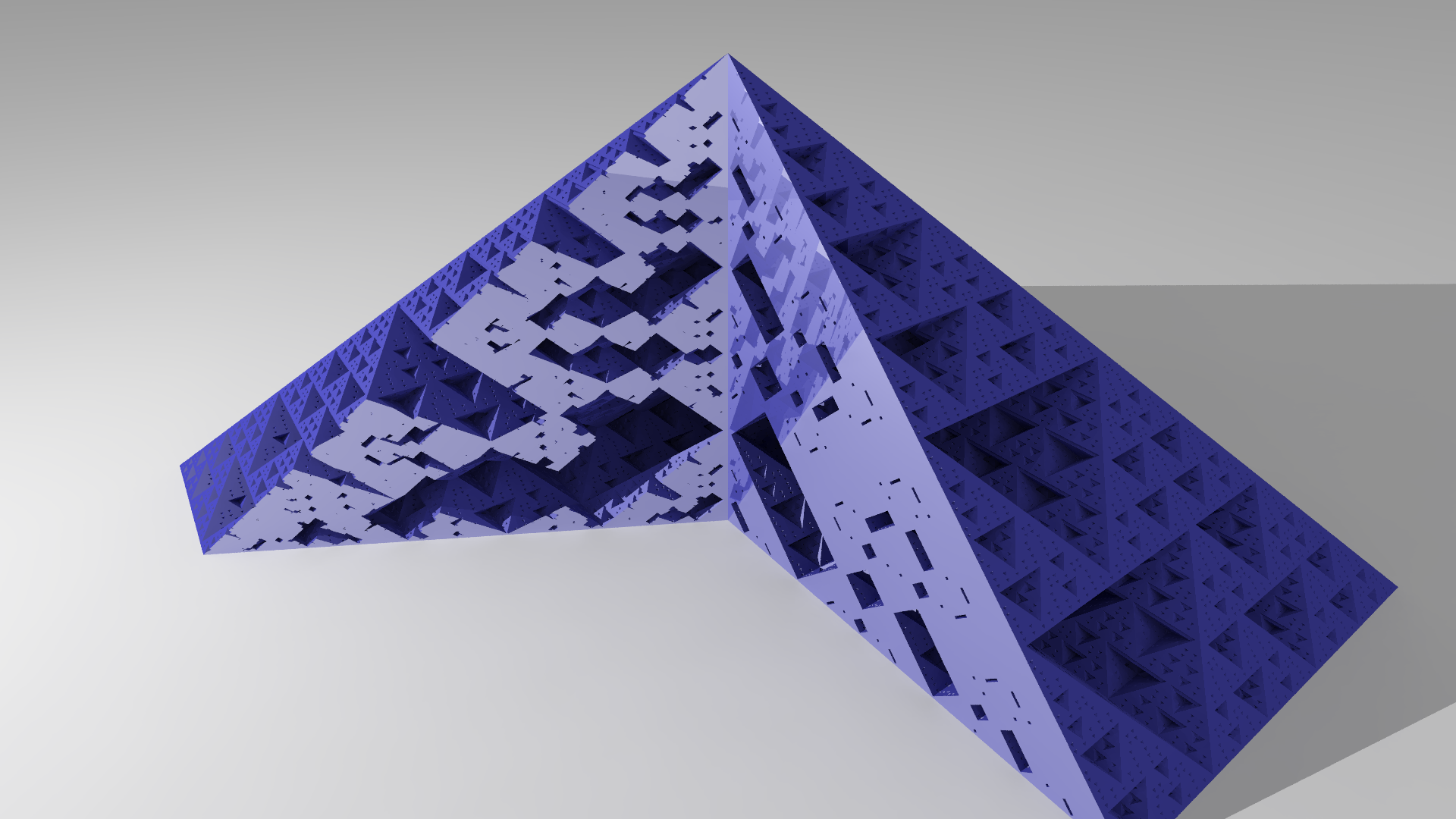}}
\caption{\label{fig:fractal3d}A modern funeral labyrinth.}
\end{figure}

For any positive integer $n$, the coefficient in $v^n$ of $\Omega^n$, seen as an element of $\joliM_2\pa{\Z_2[u,u^{-1}]}[v,v^{-1}]$, is $\pa{\begin{array}{cc}0&0\\ 0&(1+u+u^{-1})^n\end{array}}$.  Therefore, the pattern on the face of the pyramids is given by $(1+u+u^{-1})\in\joliM_1\pa{\Z_2[u,u^{-1}]}$, as can be indeed noticed on the figure, where half of this symmetric pattern is clearly visible on the external face.  By contrast, the coefficient of $\Omega^n$ in $v^0$ is exactly $\Theta^n$, which explains why the cut along $(1,0)$ in Figure~\ref{fig:fractal3d} looks like the half of Figure~\ref{fig:sub_theta} (turned upside-down).

Here are some open questions and possible future developments:

 \begin{itemize}

\item Presumably, the $m$ in Theorem~\ref{bigthm} can always be taken to be $1$. This is known to be true in the cyclic case, i.e. when $d=1$, cf \cite{selfsimilarity}.

\item The algorithm presented in this article, producing a description of the spacetime diagram in the form of a matrix substition system, has a high complexity, due to the large size of its output.  Is this a necessary evil, or can more efficient descriptions be found? Could for instance the more elegant triangle-based substitution scheme presented in section \ref{section_special_scheme} be naturally generalized?

\item To what extent can the algebraic structure be weakened? Instead of the alphabet being an Abelian group, could we consider an Abelian monoid? Is it possible to get rid of commutativity and/or associativity? 
\end{itemize}

\section*{Acknowledgements}

We would like to thank Jean-Paul Allouche, Bruno Durand and Cris Moore for their useful feedback and bibliographical hints.  We gratefully acknowledge the support of the Deutsche Forschungsgemeinschaft (Forschergruppe 635), the EU (projects CORNER and QICS), the Erwin Schr{\"o}dinger Institute and the Rosa Luxemburg Foundation.

\bibliography{frac}

\end{document}